%% file: ms.tex
\documentclass{llncs}
\usepackage{makeidx}  

\usepackage{amssymb,amsmath,amsfonts}
\usepackage{pifont}
\usepackage{graphicx}
\usepackage{latexsym}
\usepackage{todonotes}
\usepackage{xspace}
\usepackage{tikz,pgf}
\usepackage{paralist}
\usepackage{xcolor}
\usepackage{wrapfig}
\usepackage[misc]{ifsym}

\usetikzlibrary{automata,calc,arrows,shapes,backgrounds}

\definecolor{ly}{RGB}{250,250,200}
\definecolor{lg}{RGB}{200,250,200}

\input{macros}

\begin{document}
\title{Probabilistic Timed Automata with Clock-Dependent Probabilities}
\titlerunning{Probabilistic Timed Automata with Clock-Dependent Probabilities}  
%
\author{
Jeremy Sproston\textsuperscript{(\Letter)}
}
\authorrunning{Jeremy Sproston} 
%
\tocauthor{
Jeremy Sproston
}
\institute{
Dipartimento di Informatica,
University of Turin, Italy\\
\email{sproston@di.unito.it}
}

\maketitle              

\begin{abstract}
Probabilistic timed automata are classical timed automata 
extended with discrete probability distributions over edges.
We introduce clock-dependent probabilistic timed automata, 
a variant of probabilistic timed automata in which 
transition probabilities can depend linearly on clock values.
Clock-dependent probabilistic timed automata
allow the modelling of a continuous relationship between time passage and 
the likelihood of system events.
We show that the problem of deciding whether the 
maximum probability of reaching a certain location 
is above a threshold is undecidable
for clock-dependent probabilistic timed automata.
On the other hand,
we show that the maximum and minimum probability of 
reaching a certain location in clock-dependent probabilistic timed automata
can be approximated using a region-graph-based approach.
\end{abstract}
%

\section{Introduction}\label{sec:intro}
\input{intro}

\section{Clock-Dependent Probabilistic Timed Automata}\label{sec:cdpta}
\input{cdpta}

\section{Undecidability of Maximal Reachability for cdPTAs}\label{sec:undec}
\input{undec}

\section{Approximation of Reachability Probabilities}\label{sec:approx}
\input{approx}

\section{Conclusion}\label{sec:concl}
\input{concl}

\bibliographystyle{abbrv}
\bibliography{cdpta}

\appendix
\section{Proof of Proposition~\ref{prop:approx}}\label{sec:appendix}
\input{appendix}

\end{document}

%% file: macros.tex
\def\Nset{\mathbb{N}}

\def\Qset{\mathbb{Q}}
\def\Rset{\mathbb{R}}

\newcommand{\true}{\mbox{\tt true}}

\def\ra{\rightarrow}

\newcommand{\sat}{\models}

\newcommand{\comment}[1]{}

\newcommand{\setsep}{:}
\newcommand{\suchthat}{\mbox{ s.t. }}

\newcommand{\xmark}{\text{\ding{55}}}

\newcommand{\IGNORE}[1]{}

\newcommand{\nnr}{\Rset_{\geq 0}}

\newcommand{\dist}{\mathsf{Dist}}
\newcommand{\adist}{\mu}

\newcommand{\support}{\mathsf{support}}

\newcommand{\dirac}[1]{\{ {#1} \mapsto 1 \}}

\newcommand{\distsum}{\bigoplus}
\newcommand{\ctweight}{\lambda}

\newcommand{\aPTS}{\mathcal{T}}
\newcommand{\states}{S}
\newcommand{\astate}{s}
\newcommand{\anotherstate}{t}
\newcommand{\sinit}{\overline{\astate}}
\newcommand{\actions}{\mathit{Act}}
\newcommand{\anaction}{a}
\newcommand{\ptstrans}{\Delta}

\newcommand{\aninfrun}{r}
\newcommand{\afinrun}{r}
\newcommand{\infruns}[1]{\mathit{InfRuns}^{#1}}
\newcommand{\finruns}[1]{\mathit{FinRuns}^{#1}}

\newcommand{\last}{\mathit{last}}

\newcommand{\infrunsstate}[2]{\mathit{InfRuns}^{#1}({#2})}

\newcommand{\astrat}{\sigma}
\newcommand{\strategies}{\Sigma}

\newcommand{\astratset}{\strategies'}

\newcommand{\altstrategies}[1]{\strategies_{#1}}
\newcommand{\altstrategiespts}[2]{\strategies^{#1}_{#2}}

\newcommand{\prob}[2]{\mathrm{Pr}^{#1}_{#2}}

\newcommand{\anequiv}{\equiv}
\newcommand{\anec}{C}
\newcommand{\aprobsim}{\preceq}

\newcommand{\actset}{A}
\newcommand{\astateset}{\states'}

\newcommand{\clocks}{\mathcal{X}}
\newcommand{\aclock}{x}
\newcommand{\anotherclock}{y}

\newcommand{\aval}{v}
\newcommand{\valuations}{\nnr^{\clocks}}
\newcommand{\aclockset}{X}

\newcommand{\pedges}{\mathit{prob}}
\newcommand{\apedge}{p}
\newcommand{\pd}{\mathfrak{p}}

\newcommand{\locs}{L}
\newcommand{\locinit}{\bar{l}}

\newcommand{\inv}{\mathit{inv}}

\newcommand{\aloc}{l}

\newcommand{\g}{g}

\newcommand{\semPTS}[1]{[\![{#1}]\!]}

\newcommand{\bzero}{\mathbf{0}}
\newcommand{\val}{v}
\newcommand{\vinz}{\sat}

\newcommand{\cconparam}[1]{\mathit{CC}({#1})}
\newcommand{\ccons}{\cconparam{\clocks}}
\newcommand{\cc}{\psi}

\newcommand{\reset}[2]{\mathsf{Reset}({#1},{#2})}

\newcommand{\maxval}[3]{\mathbb{P}^\mathrm{max}_{{#1},{#3}}({#2})}
\newcommand{\minval}[3]{\mathbb{P}^\mathrm{min}_{{#1},{#3}}({#2})}

\newcommand{\Sup}[1]{\sup_{#1}}
\newcommand{\Inf}[1]{\inf_{#1}}

\newcommand{\adelay}{\delta}

\newcommand{\finalstates}{\states_F}
\newcommand{\pthresh}{\lambda}

\newcommand{\adisttemp}{\mathfrak{d}}
\newcommand{\adtpara}[2]{{#1}[{#2}]}
\newcommand{\disttemps}{\mathfrak{Dist}}

\newcommand{\acdpta}{\mathcal{P}}

\newcommand{\ptstransdelay}{\overrightarrow{\ptstrans}}
\newcommand{\ptstranspedge}{\widehat{\ptstrans}}

\newcommand{\finallocs}{F}

\newcommand{\cdptastrategies}{\mathbf{\strategies}}

\newcommand{\interval}[2]{I^{#1}_{#2}}

\newcommand{\apartition}[2]{\mathcal{I}^{#1}_{#2}}
\newcommand{\aninterval}{I}

\newcommand{\outcome}{e}
\newcommand{\plf}[2]{f^{#1}_{#2}}
\newcommand{\fc}[3]{c^{#1}_{{#2},{#3}}}
\newcommand{\fd}[3]{d^{#1}_{{#2},{#3}}}

\newcommand{\twocm}{\mathcal{M}}
\newcommand{\instrs}{\mathcal{L}}
\newcommand{\counters}{\mathcal{C}}
\newcommand{\anin}{\ell}
\newcommand{\acount}{\mathsf{c}}
\newcommand{\avalc}{\mathsf{v}}

\newcommand{\floor}[1]{\lfloor{#1}\rfloor}

\newcommand{\seq}[2]{\langle {#1} \rangle_{{#2}}}
\newcommand{\set}[2]{\{ {#1} \}_{{#2}}}

\newcommand{\maxbound}{M}
\newcommand{\anatval}{h}
\newcommand{\acp}{\alpha}
\newcommand{\acpconst}{a}

\newcommand{\areg}{R}

\newcommand{\multof}[1]{[{#1}]}

\newcommand{\cps}[2]{\mathsf{CornerPoints}_{#2}}
\newcommand{\cpsofreg}[1]{\mathsf{CP}({#1})}

\newcommand{\regs}[2]{\mathsf{Regs}_{#2}}

\newcommand{\acdrg}[2]{\mathcal{A}_{#2}}
\newcommand{\rgstates}[2]{\mathsf{S}_{#2}}
\newcommand{\rginit}{\overline{\mathsf{s}}}
\newcommand{\rgactions}[2]{\mathsf{Act}_{#2}}
\newcommand{\rgtrans}[2]{\Gamma_{#2}}

\newcommand{\rsuccact}{\tau}

\newcommand{\rgtransdelay}[2]{\overrightarrow{\rgtrans{#1}{#2}}}
\newcommand{\rgtranspedge}[2]{\widehat{\rgtrans{#1}{#2}}}

\newcommand{\anrgdist}{\nu}

\newcommand{\rgstrategies}[2]{\mathbf{\Pi}_{#2}}

\newcommand{\finalregs}[2]{\mathsf{Regs}_{#2}^F}

\newcommand{\inrg}[4]{\langle \! [ {#1},{#2} ] \! \rangle_{#4}}

\newcommand{\weight}{\theta}

\newcommand{\proposedpsim}{\aprobsim}

\newcommand{\regcont}[2]{[{#1}]_{#2}}



%% file: intro.tex
Reactive systems are increasingly required to satisfy a combination
of qualitative criteria (such as safety and liveness)
and quantitative criteria (such as timeliness, reliability and performance).
This trend has led to the development of techniques and tools
for the formal verification of both qualitative and quantitative properties.
In this paper, we consider a formalism for real-time systems that exhibit randomised behaviour,
namely probabilistic timed automata (PTA) \cite{GJ95,KNSS02}.
PTAs extend classical Alur-Dill timed automata \cite{AD94}
with discrete probabilistic branching over automata edges; 
alternatively a PTA can be viewed as a Markov decision process \cite{Put94}
or a Segala probabilistic automaton \cite{SegPhD}
extended with timed-automata-like clock variables and constraints over those clocks.
PTAs have been used previously to model case studies
including randomised protocols and scheduling problems with uncertainty \cite{KNPS06,NPS13},
some of which have become standard benchmarks in the field of probabilistic model checking.

We recall briefly the behaviour of a PTA:
as time passes, the model stays within a particular discrete state,
and the values of its clocks increase at the same rate;
at a certain point in time, 
the model can leave the discrete state if the current values of the clocks satisfy 
a constraint (called a guard) labelling one of the probability distributions over edges leaving the state;
then a probabilistic choice as to which discrete state 
to then visit is made according to the chosen edge distribution.
In the standard presentation of PTAs,
any dependencies between time and probabilities over edges 
must be defined by utilising multiple distributions enabled with different sets of clock values.
For example, to model the fact that a packet loss is more likely as time passes,
we can use clock $\aclock$ to measure time, 
and two distributions $\adist_1$ and $\adist_2$ assigning probability $\lambda_1$ and $\lambda_2$
(for $\lambda_1 < \lambda_2$), respectively, 
to taking edges leading to a discrete state corresponding to packet loss,
where the guard of $\adist_1$ is $\aclock \leq c$
and the guard of $\adist_2$ is $\aclock > c$, for some constant $c \in \Nset$.
Hence, when the value of clock $\aclock$ is not more than $c$, 
a packet loss occurs with probability $\lambda_1$,
otherwise it occurs with probability $\lambda_2$.
A more direct way of expressing the relationship between time and probability 
would be letting the probability of making a transition to a discrete state representing packet loss
be dependent on the value of the clock,
i.e., let the value of this probability be equal to $f(\aclock)$,
where $f$ is an increasing function from the values of $\aclock$ to probabilities.
We note that such a kind of dependence of discrete branching probabilities 
on values of continuous variables is standard in the field of stochastic hybrid systems,
for example in \cite{AKLP10}.

In this paper we consider such a formalism based on PTAs,
in which all probabilities used by edge distributions 
can be expressed as functions of values of the clocks used by the model:
the resulting formalism is called \emph{clock-dependent probabilistic timed automata} (cdPTA).
We focus on a simple class of functions from clock values to probabilities,
namely those that can be expressed as sums of continuous piecewise linear functions,
%
and consider a basic problem in the context of probabilistic model checking,
namely probabilistic reachability:
determine whether the maximum (respectively, minimum) probability of reaching a certain set of locations
from the initial state is above (respectively, below) a threshold.
After introducing cdPTAs (in Section \ref{sec:cdpta}),
our first result (in Section \ref{sec:undec}) is that the probabilistic reachability problem is undecidable 
for cdPTA with a least three clocks.
This result is inspired from recent related work on stochastic timed Markov decision processes \cite{ABKMT16}.
Furthermore, we give an example of cdPTA with one clock
for which the maximal probability of reaching a certain location
involves a particular edge being taken when the clock has an irrational value.
This suggests that classical techniques for partitioning the state space into a finite number of equivalence classes
on the basis of a fixed, rational-numbered time granularity, 
such as the region graph \cite{AD94} or the corner-point abstraction \cite{BBL08},
cannot be applied directly to the case of cdPTA
to obtain optimal reachability probabilities, 
because they rely on the fact that optimal choices can be made either at or arbitrarily closely to 
clock values that are multiples of the chosen rational-numbered time granularity.
In Section \ref{sec:approx},
we present a conservative approximation method for cdPTA,
i.e., maximum (respectively, minimum) probabilities are 
bounded from above (respectively, from below) in the approximation.
This method is based on the region graph 
but uses concepts from the corner-point abstraction to define transition distributions.
%
%
We show that successive refinement of the approximation,
obtained by increasing the time granularity by a constant factor,
does not lead to a more conservative approximation:
in practice, in many cases such a refinement can lead to 
a substantial improvement in the computed probabilities,
which we show using a small example.

%
%

%% file: cdpta.tex
\subsubsection{Preliminaries.}

We use $\nnr$ to denote the set of non-negative real numbers,
$\Qset$ to denote the set of rational numbers and
$\Nset$ to denote the set of natural numbers.
A (discrete) probability \emph{distribution} over a countable set $Q$ is
a function $\adist: Q \ra [0,1]$ such that $\sum_{q \in Q} \adist(q) = 1$.
For a function $\adist : Q \ra \nnr$ we define $\support(\adist) = \{q \in Q \setsep \adist(q) > 0 \}$.
Then for an uncountable set $Q$ we define $\dist(Q)$ to be the set of
functions $\adist : Q \ra [0,1]$, such that $\support(\adist)$ is a
countable set and $\adist$ restricted to $\support(\adist)$ is a (discrete)
probability distribution.
Given $q \in Q$, we use $\dirac{q}$ to denote the distribution 
that assigns probability 1 to the single element $q$.

A \emph{probabilistic transition system} (PTS)
$\aPTS = ( \states, \sinit, \actions, \ptstrans )$ comprises the following components:
a set $\states$ of \emph{states} with an \emph{initial state} $\sinit \in \states$,
a set $\actions$ of \emph{actions},
and a \emph{probabilistic transition relation}
$\ptstrans \subseteq \states \times \actions \times \dist(\states)$.
The sets of states, actions and the probabilistic transition relation
can be uncountable.
Transitions from state to state of a PTS are performed in two
steps: if the current state is $\astate$, the first step concerns a
nondeterministic selection of a probabilistic transition $(\astate,\anaction,\adist) \in \ptstrans$; 
the second step comprises a probabilistic choice, made according to the distribution $\adist$, 
as to which state to make the transition 
(that is, a transition to a state $\astate' \in \states$ is made with probability $\adist(\astate')$).
We denote such a completed transition by $\astate \xrightarrow{\anaction,\adist} \astate'$.
%
We assume that 
for each state $\astate \in \states$ there exists some $(\astate,\anaction,\adist) \in \ptstrans$.
%

An \emph{infinite run} of the PTS $\aPTS$
is an infinite sequence of consecutive transitions
$\aninfrun = \astate_0 \xrightarrow{\anaction_0,\adist_0} \astate_1 \xrightarrow{\anaction_1,\adist_1} \cdots$
(i.e., the target state of one transition is the source state of the next).
Similarly, a \emph{finite run} of $\aPTS$
is a finite sequence of consecutive transitions
$\afinrun = 
\astate_0 \xrightarrow{\anaction_0,\adist_0} \astate_1 \xrightarrow{\anaction_1,\adist_1} \cdots \xrightarrow{\anaction_{n-1},\adist_{n-1}} s_n$.
We use $\infruns{\aPTS}$ to denote the set of infinite runs of $\aPTS$,
and $\finruns{\aPTS}$ the set of finite runs of $\aPTS$.
If $\afinrun$ is a finite run, we denote by $\last(\afinrun)$ the last state of $\afinrun$.
For any infinite run $\afinrun$ and $i \in \Nset$, 
let $\afinrun(i)=\astate_i$ be the $(i+1)$th state along $\afinrun$.
Let $\infrunsstate{\aPTS}{\astate}$ 
refer to the set of infinite runs of $\aPTS$
commencing in state $\astate \in \states$.

A \emph{strategy} of a PTS $\aPTS$ is a function $\astrat$
mapping every finite run $\afinrun \in \finruns{\aPTS}$ to a 
distribution in $\dist(\ptstrans)$
such that $(\astate,\anaction,\adist) \in \support(\astrat(\afinrun))$ implies that 
$\astate = \last(\afinrun)$.
From \cite[Lemma~4.10]{Hah13},
without loss of generality we can assume henceforth that strategies
map to distributions assigning positive probability to finite sets of elements,
i.e., strategies $\astrat$ for which $|\support(\astrat(\afinrun))|$ is finite for all $\afinrun \in \finruns{\aPTS}$.
For any strategy 
$\astrat$,
let $\infruns{\astrat}$ denote the set of infinite runs
resulting from the choices of $\astrat$.
For a state $\astate \in \states$,
let $\infrunsstate{\astrat}{\astate} = \infruns{\sigma} \cap \infrunsstate{\aPTS}{\astate}$.
Given a strategy 
$\astrat$ and a state $\astate \in \states$,
we define the probability measure $\prob{\astrat}{\astate}$ over $\infrunsstate{\astrat}{\astate}$
in the standard way \cite{KSK66}.

Given a set $\finalstates \subseteq \states$,
define $\Diamond \finalstates  
= \{ \aninfrun \in \infruns{\aPTS} \setsep \exists i \in \Nset \suchthat \afinrun(i) \in \finalstates  \}$
to be the set of infinite runs of $\aPTS$ such that some state of $\finalstates$ is visited along the run.
Given a set $\astratset \subseteq \strategies$ of strategies,
we define the 
\emph{maximum value over $\astratset$ with respect to $\finalstates$} as
$\maxval{\aPTS}{\finalstates}{\astratset} 
= 
\Sup{\astrat \in \astratset}~ \prob{\astrat}{\sinit}(\Diamond \finalstates)$.
Similarly, 
the \emph{minimum value over $\astratset$ with respect to $\finalstates$} 
is defined as
$\minval{\aPTS}{\finalstates}{\astratset} 
= 
\Inf{\astrat \in \astratset}~ \prob{\astrat}{\sinit}(\Diamond \finalstates)$.
%
%
The \emph{maximal reachability problem} 
for $\aPTS$, $\finalstates \subseteq \states$, $\astratset \subseteq \strategies$, 
$\unrhd \in \{ \geq, > \}$ and $\pthresh \in [0,1]$
is to decide whether $\maxval{\aPTS}{\finalstates}{\astratset} \unrhd \pthresh$.
Similarly, the \emph{minimal reachability problem}
for $\aPTS$, $\finalstates \subseteq \states$, $\astratset \subseteq \strategies$,
$\unlhd \in \{ \leq, < \}$ and $\pthresh \in [0,1]$
is to decide whether $\minval{\aPTS}{\finalstates}{\astratset} \unlhd \pthresh$.
%


\subsubsection{Clock-Dependent Probabilistic Timed Automata.}

Let $\clocks$ be a finite set of
real-valued variables called \emph{clocks}, 
the values of which increase at the same rate as real-time
and which can be reset to 0.
A function $\aval: \clocks \rightarrow \nnr$ is referred to as a \emph{clock valuation}
and the set of all clock valuations is denoted by $\valuations$.
For $\aval \in \valuations$, $t \in \nnr$ and $\aclockset \subseteq \clocks$, 
we use $\aval{+}t$ to denote the clock valuation that increments all clock values in $\aval$ by $t$,
and $\aval[\aclockset{:=}0]$ to denote the clock valuation in which clocks in $\aclockset$ are reset to 0.

For a set $Q$, 
a \emph{distribution template} $\adisttemp : \valuations \ra \dist(Q)$
gives a distribution over $Q$ for each clock valuation.
In the following, we use notation $\adtpara{\adisttemp}{\aval}$, rather than $\adisttemp(\aval)$,
to denote the distribution corresponding to distribution template $\adisttemp$ and clock valuation $\aval$.
Let $\disttemps(Q)$ be the set of distribution templates over $Q$.

The set $\ccons$ of {\em clock constraints} over $\clocks$ is defined
as the set of conjunctions over atomic formulae of the
form
$\aclock \sim c$, where $\aclock \in \clocks$,
$\sim \in \{ <,\leq,\geq,> \}$, and $c \in \Nset$.
A clock valuation $\aval$ satisfies a clock constraint $\cc$,
denoted by $\aval \vinz \cc$, 
if $\cc$ resolves to $\true$ when substituting each occurrence of clock $x$ with $\aval(x)$.
\sloppypar{
A \emph{clock-dependent probabilistic timed automaton} (cdPTA)  
$\acdpta = (\locs, \locinit, \clocks, \inv, \pedges)$
comprises the following components:
a finite set
$\locs$ of \emph{locations} with an \emph{initial location} $\locinit \in \locs$;
a finite set $\clocks$ of clocks;
a function $\inv: \locs \ra \ccons$ associating an \emph{invariant condition} with each location;
a set $\pedges \subseteq \locs \times \ccons \times \disttemps(2^{\clocks} \times \locs)$
of \emph{probabilistic edges}.
A probabilistic edge $(\aloc,\g,\pd) \in \pedges$ comprises:
\begin{inparaenum}[(1)]
\item a source location $\aloc$;
\item a clock constraint $\g$, called a \emph{guard}; and
\item a distribution template $\pd$ 
with respect to pairs of the form $(\aclockset,\aloc') \in 2^{\clocks} \times \locs$ 
(i.e., pairs consisting of a set $\aclockset$ of clocks to be reset and a target location $\aloc'$).
\end{inparaenum}
}
The behaviour of a cdPTA takes a similar form to
that of a standard probabilistic timed automaton \cite{GJ95,KNSS02}:
in any location time can advance as long as the invariant holds,
and the choice as to how much time elapses is made nondeterministically;
a probabilistic edge can be taken if its guard is satisfied
 by the current values of the clocks and,
again, the choice as to which probabilistic edge to take is made nondeterministically;
for a taken probabilistic edge, the choice of which clocks to reset and which target
location to make the transition to is \emph{probabilistic}.
The key difference with cdPTAs is that the distribution used to make this probabilistic choice
depends on the probabilistic edge taken \emph{and}
on the current clock valuation.

\begin{figure}[t]
{

\centering
\scriptsize

\begin{tikzpicture}[->,>=stealth',shorten >=1pt,auto, thin] 

  \tikzstyle{state}=[draw=black, text=black, shape=circle, inner sep=3pt, outer sep=0pt, circle, rounded corners, fill=ly] 
  \tikzstyle{final_state}=[draw=black, text=black, shape=circle, inner sep=3pt, outer sep=0pt, circle, rounded corners, fill= lg] 

	\node(for_initial){};

	\node[state, node distance=1cm](TL)[right of=for_initial]{$\mathrm{TL}$};
	\node[yshift=-0.3cm,above of = TL](inv_TL){\begin{tabular}{c}{$\aclock \leq 2$}\\{$\anotherclock \leq c_\mathrm{max}$}\end{tabular}};

	\node[fill=black, node distance=2.5cm] (nail_TLr) [right of=TL] {};
	\node[fill=black, node distance=1.2cm] (nail_TLd) [below of=TL] {};

	\node[state, node distance=2.5cm](TR)[right of=nail_TLr]{$\mathrm{TR}$};
	\node[yshift=-0.3cm,above of = TR](inv_TR){\begin{tabular}{c}{$\aclock \leq 2$}\\{$\anotherclock \leq c_\mathrm{max}$}\end{tabular}};

	\node[fill=black, node distance=1.2cm] (nail_TR) [below of=TR] {};

	\node[state, node distance=1.2cm](BL)[below of=nail_TLd]{$\mathrm{BL}$};
	\node[yshift=0.3cm,below of = BL](inv_BL){\begin{tabular}{c}{$\aclock \leq 3$}\\{$\anotherclock \leq c_\mathrm{max}$}\end{tabular}};

	\node[fill=black, node distance=2.5cm] (nail_BL) [right of=BL] {};

	\node[state, node distance=1.2cm](BR)[below of=nail_TR]{$\mathrm{BR}$};
	\node[yshift=0.5cm,below of = BR](inv_BR){$\aclock \leq 0$};

	\node[fill=black, node distance=2.5cm] (nail_BR) [right of=BR] {};

	\node[final_state, node distance=1.2cm](F)[above of=nail_BR]{$\xmark$};
	\node[final_state, node distance=2cm,accepting](S)[right of=nail_BR]{$\checkmark$};
	
	\node[node distance=1.7cm](NW)[above left of=TL]{};
	\node[node distance=1.7cm](NE)[right of=TR]{};
	\node[node distance=1.7cm](SW)[below left of=BL]{};
	\node[node distance=1.7cm](SE)[below right of=BR]{};

		
	\path[rounded corners]
	(for_initial)
		edge (TL)

	(TL)
		edge [above, bend right] node {$\aclock \geq 1$} (nail_TLr)
		edge [right, bend left] node {$\aclock \geq 1$} (nail_TLd)

	(nail_TLr)
		edge [above] node {\colorbox{gray!15}{$\aclock-1$}} (TR)
		edge [below] node {$\{ \aclock \}$} (TR)
		edge [above, bend right] node {\begin{tabular}{c}{\colorbox{gray!15}{$2-\aclock$}}\\{$\{ \aclock \}$}\end{tabular}} (TL)

	(nail_TLd)
		edge [right] node {$\{ \aclock \}$} (BL)
		edge [left] node {\colorbox{gray!15}{$\aclock-1$}} (BL)
		edge [left, bend left] node {\begin{tabular}{c}{\colorbox{gray!15}{$2-\aclock$}}\\{$\{ \aclock \}$}\end{tabular}} (TL)

	(TR)
		edge [left, bend right] node {$\aclock \geq 1$} (nail_TR)

	(nail_TR)
		edge [right] node {$\{ \aclock \}$} (BR)
		edge [left] node {\colorbox{gray!15}{$\frac{\aclock-1}{2}$}} (BR)
		edge [right, bend right] node {\begin{tabular}{c}{\colorbox{gray!15}{$1-\frac{\aclock-1}{2}$}}\\{$\{ \aclock \}$}\end{tabular}} (TR)

	(BL)
		edge [below, bend left] node {$\aclock \geq 2$} (nail_BL)

	(nail_BL)
		edge [above] node {\colorbox{gray!15}{$\aclock-2$}} (BR)
		edge [below] node {$\{ \aclock \}$} (BR)
		edge [below, bend left] node {\begin{tabular}{c}{\colorbox{gray!15}{$3-\aclock$}}\\{$\{ \aclock \}$}\end{tabular}} (BL)

	(BR)
		edge [above] node {$\aclock = 0$} (nail_BR)

	(nail_BR)
		edge [below] node {\colorbox{gray!15}{$1-\frac{\anotherclock}{c_\mathrm{max}}$}} (S)
		edge [right] node {\colorbox{gray!15}{$\frac{\anotherclock}{c_\mathrm{max}}$}} (F)

	(NE)
		edge [right,bend left,bend angle=60] node {$\anotherclock = c_\mathrm{max}$} (F);

  \begin{pgfonlayer}{background}
    \draw [dashed,join=round]
      (NW.north  -| NE.east)  rectangle (SE.south  -| SW.west);
  \end{pgfonlayer}

\end{tikzpicture}

}

\caption{A cdPTA modelling a simple robot example.}
\label{fig:robot}
\end{figure}
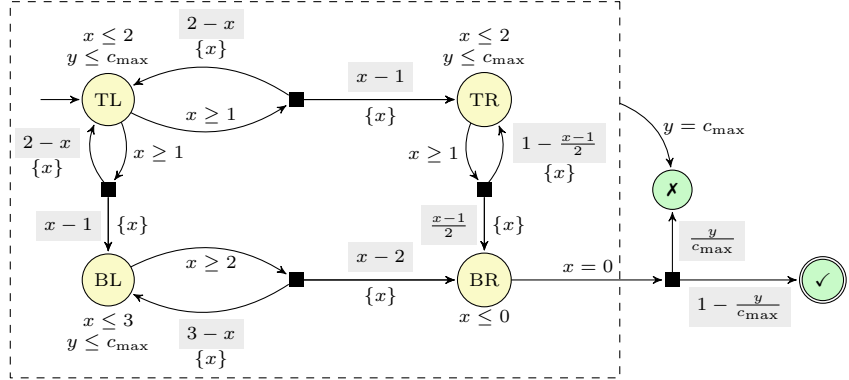

\paragraph{Example 1.}
In Figure~\ref{fig:robot} we give an example of a cdPTA modelling a simple robot 
that must reach a certain geographical area and then carry out a particular task.
The usual conventions for the graphical representation of timed automata are used in the figure.
Black squares denote the distributions of probabilistic edges,
and expressions on probabilities used by distribution templates are written 
with a grey background on their outgoing arcs. 
The robot can be in one of four geographical areas, 
which can be thought of as cells in a $2 \times 2$ grid,
each of which corresponds to a cdPTA location.
The robot begins in the top-left cell (corresponding to location $\mathrm{TL}$),
and its objective is to reach the bottom-right cell (location $\mathrm{BR}$).
The robot can move either to the top-right cell (location $\mathrm{TR}$), 
or to the bottom-left cell (location $\mathrm{BL}$),
then to the bottom-right cell.
In each cell, the robot must wait a certain amount of time 
($1$ time units in the top cells and $2$ time units in the bottom-left cell)
before attempting to leave the cell
(for example, to recharge solar batteries),
after which it can spend at most $1$ time unit attempting to leave the cell.
With a certain probability, the attempt to leave the cell will fail,
and the robot must wait before trying to leave the cell again;
the more time is dedicated to leaving the cell, the more likely the robot will succeed.
Although passing through the top-right cell is not slower than passing through the bottom-left cell,
the probability of leaving the cell successfully increases at a slower rate than in other cells
(representing, for example, terrain in which the robot finds it difficult to navigate).
On arrival in the bottom-right cell, 
the robot successfully carries out its task with a probability that is inversely proportional to the total time elapsed
(for example, 
the robot could be transporting medical supplies,
the efficacy of which may be inversely proportional to the time elapsed).
The clock $\aclock$ is used to represent the amount of time used by the robot in its attempt 
to move from cell to cell,
whereas
the clock $\anotherclock$ represents the total amount of time since the start of the robot's mission.
If the clock $\anotherclock$ reaches its maximum amount $c_\mathrm{max}$,
then the mission fails (as denoted by the edge to the location denoted by $\xmark$,
which is available in locations $\mathrm{TL}$, $\mathrm{TR}$, $\mathrm{BL}$ and $\mathrm{BR}$,
as indicated by the dashed box).
The objective of the robot's controller is to maximise the probability of reaching the location denoted by $\checkmark$.
Note that there is a trade-off between dedicating more time to movement between the cells,
which increases the probability of successful navigation and therefore progress towards the target point,
and spending less time on the overall mission,
which increases the probability of carrying out the required task at the target point.
\qed

A \emph{state\/} of a cdPTA is a pair comprising a location and a clock valuation
satisfying the location's invariant condition,
i.e., $(\aloc,\aval) \in \locs \times \valuations$ such that $\aval \vinz \inv(\aloc)$.
In any state $(\aloc,\aval)$, either a certain amount of time $\adelay \in \nnr$ elapses, 
or a probabilistic edge is traversed.
If time elapses, then
the choice of $\adelay$ requires that the invariant $\inv(\aloc)$ remains continuously satisfied while time passes. 
The resulting state after this transition is $(\aloc,\aval{+}\adelay)$.
A probabilistic edge $(\aloc',\g,\pd) \in \pedges$ can be chosen from $(\aloc,\aval)$ if $\aloc = \aloc'$ and it is \emph{enabled},
i.e., the clock constraint $\g$ is satisfied by $\aval$.
Once a probabilistic edge $(\aloc,\g,\pd)$ is chosen, a set of clocks to reset and a successor location are selected at random,
according to the distribution $\adtpara{\pd}{\aval}$.

We make a number of assumptions concerning the cdPTA models considered.
Firstly, we restrict our attention to cdPTAs for which it is always possible to take a probabilistic edge, 
either immediately or after letting time elapse.
This condition holds generally for PTA models in practice \cite{KNPS06}.
A sufficient syntactic condition for this property has been presented formally in \cite{JLS08}.
Secondly, 
we consider cdPTAs that feature invariant conditions that prevent clock values from exceeding some bound:
formally, for each location $\aloc \in \locs$, we have that $\inv(\aloc)$ 
contains a constraint of the form $\aclock \leq c$ or $\aclock < c$ for each clock $\aclock \in \clocks$.
Thirdly, we assume that all possible target states of probabilistic edges satisfy their invariants:
for all probabilistic edges $(\aloc,\g,\pd) \in \pedges$,
for all clock valuations $\aval \in \valuations$ such that $\aval \vinz \g$,
and for all $(\aclockset,\aloc') \in 2^\clocks \times \locs$, 
we have that $\adtpara{\pd}{\aval}(\aclockset,\aloc')>0$ implies $\aval[\aclockset:=0] \vinz \inv(\aloc')$.
%
%
Finally, we assume that any clock valuation that satisfies the guard of a probabilistic edge also 
satisfies the invariant of the source location:
this can be achieved, without changing the underlying semantic PTS, 
by replacing each probabilistic edge $(\aloc,\g,\pd) \in \pedges$ by $(\aloc,\g \wedge \inv(\aloc),\pd)$.

Let $\bzero \in \valuations$ be the clock valuation which assigns 0 to all
clocks in $\clocks$.
The semantics of the cdPTA
$\acdpta = (\locs, \locinit, \clocks, \inv, \pedges)$
is the PTS
$\semPTS{\acdpta} = ( \states, \sinit, \actions, \ptstrans )$
where:
\begin{itemize}
\item
$\states = \{ (\aloc,\aval) \setsep \aloc \in \locs \mbox{ and } \aval \in \nnr^\clocks \suchthat \aval \vinz \inv(\aloc) \}$
and $\sinit = \{ (\locinit,\bzero) \}$;
\item
$\actions = \nnr \cup \pedges$;
\item
$\ptstrans = \ptstransdelay \cup \ptstranspedge$,
where $\ptstransdelay \subseteq \states \times \nnr \times \dist(\states)$ 
and $\ptstranspedge \subseteq \states \times \pedges \times \dist(\states)$ such that:
\begin{itemize}
\item
$\ptstransdelay$ is the smallest set such that 
$((\aloc,\aval),\adelay,\dirac{(\aloc,\aval+\adelay)}) \in \ptstransdelay$ 
if there exists $\adelay \in \nnr$
such that $\aval+\adelay' \vinz \inv(\aloc)$ for all $0 \leq \adelay' \leq \adelay$;
\item
$\ptstranspedge$ is the smallest set such that 
$((\aloc,\aval),(\aloc,\g,\pd),\adist) \in \ptstranspedge$
if
\begin{enumerate}
\item
$\aval \vinz \g$;
\item
for any $(\aloc',\aval') \in \states$, we have 
$\adist(\aloc',\aval') = \sum_{\aclockset \in \reset{\aval}{\aval'}} \adtpara{\pd}{\aval}(\aclockset,\aloc')$,
where $\reset{\aval}{\aval'} = \{ \aclockset \subseteq \clocks \mid \aval[\aclockset:=0] = \aval' \}$.
\end{enumerate}
\end{itemize}
\end{itemize}

When considering maximum and minimum values for cdPTAs,
we henceforth consider strategies that alternate between transitions from $\ptstransdelay$
(time elapse transitions)
and transitions from $\ptstranspedge$
(probabilistic edge transitions).
Formally,
a \emph{cdPTA strategy} $\astrat$ is a strategy such that, 
for a finite run $\afinrun \in \finruns{\semPTS{\acdpta}}$ 
that has $\astate \xrightarrow{\anaction,\adist} \astate'$ as its final transition,
either $(\astate,\anaction,\adist) \in \ptstransdelay$ and $\support(\astrat(\afinrun)) \in \ptstranspedge$,
or $(\astate,\anaction,\adist) \in \ptstranspedge$ and $\support(\astrat(\afinrun)) \in \ptstransdelay$.
We write $\cdptastrategies$ for the set of cdPTA strategies of $\semPTS{\acdpta}$.
%
Given a set $\finallocs \subseteq \locs$ of locations, subsequently called \emph{target locations},
we let $\finalstates = \{ (\aloc,\aval) \in \states \setsep \aloc \in \finallocs\}$.
Let $\unrhd \in \{ \geq, > \}$, $\unlhd \in \{ \leq, < \}$ and $\pthresh \in [0,1]$:
then the maximal (respectively, minimal) reachability problem for cdPTA is to decide whether 
$\maxval{\semPTS{\acdpta}}{\finalstates}{\cdptastrategies} \unrhd \pthresh$
(respectively, $\minval{\semPTS{\acdpta}}{\finalstates}{\cdptastrategies} \unlhd \pthresh$).

\subsubsection{Piecewise Linear Clock Dependencies.}
In this paper, we concentrate on a particular subclass of distribution templates
based on continuous piecewise linear functions.
Let $\aclock \in \clocks$ be a clock and $\apedge = (\aloc,\g,\pd) \in \pedges$ be a probabilistic edge.
Let $\interval{\apedge}{\aclock}$ be the interval containing the values of $\aclock$ 
of clock valuations that satisfy $\g$: 
formally $\interval{\apedge}{\aclock} = 
\{ \aval(\aclock) \in \nnr \setsep \aval \in \valuations \suchthat \aval \vinz \g \}$.
For example, for $\g = (\aclock \geq 3) \wedge (\aclock < 5) \wedge (\anotherclock \leq 8)$,
we have $\interval{\apedge}{\aclock}=[3,5)$ and $\interval{\apedge}{\anotherclock}=[0,8]$.
We equip each probabilistic edge $\apedge = (\aloc,\g,\pd) \in \pedges$ 
and $\outcome = (\aclockset,\aloc') \in 2^\clocks \times \locs$
with a continuous piecewise linear function $\plf{\apedge,\outcome}{\aclock}$ 
with domain $\interval{\apedge}{\aclock}$ 
for each clock $\aclock \in \clocks$.
Formally, we consider a partition 
$\apartition{\apedge,\outcome}{\aclock}$ of $\interval{\apedge}{\aclock}$
(i.e., $\bigcup_{\aninterval \in \apartition{\apedge,\outcome}{\aclock}} \aninterval = \interval{\apedge}{\aclock}$
and $\aninterval \cap \aninterval' = \emptyset$ 
for each $\aninterval,\aninterval' \in \apartition{\apedge,\outcome}{\aclock}$ such that $\aninterval \neq \aninterval'$),
and sets $\set{\fc{\apedge,\outcome}{\aclock}{\aninterval}}{\aninterval \in \apartition{\apedge,\outcome}{\aclock}}$
and $\set{\fd{\apedge,\outcome}{\aclock}{\aninterval}}{\aninterval \in \apartition{\apedge,\outcome}{\aclock}}$ 
of constants in $\Qset$
such that:
\begin{inparaenum}[(a)]
\item
for every $\aninterval \in \apartition{\apedge,\outcome}{\aclock}$ and $\gamma \in \aninterval$,
we have 
$\plf{\apedge,\outcome}{\aclock}(\gamma) = 
\fc{\apedge,\outcome}{\aclock}{\aninterval} + \fd{\apedge,\outcome}{\aclock}{\aninterval} \cdot \gamma$;
\item
$\plf{\apedge,\outcome}{\aclock}$ is continuous
(i.e., for each $\gamma \in \interval{\apedge}{\aclock}$,
we have $\lim_{\zeta \ra \gamma} \plf{\apedge,\outcome}{\aclock}(\zeta) = \plf{\apedge,\outcome}{\aclock}(\gamma)$).
\end{inparaenum} 
We make the following assumptions for each probabilistic edge $\apedge \in \pedges$:
\begin{inparaenum}[(1)]
\item
all endpoints of intervals in $\apartition{\apedge,\outcome}{\aclock}$ are natural numbers,
for all clocks $\aclock \in \clocks$ and $\outcome \in 2^\clocks \times \locs$;
\item
$\sum_{\aclock \in \clocks} \plf{\apedge,\outcome}{\aclock}(\aval(\aclock)) \in [0,1]$ 
for each $\outcome \in 2^\clocks \times \locs$
and $\aval \in \valuations$ such that $\aval \vinz \g$;
\item
$\sum_{\outcome \in 2^\clocks \times \locs} \sum_{\aclock \in \clocks} \plf{\apedge,\outcome}{\aclock}(\aval(\aclock)) = 1$
for each $\aval \in \valuations$ such that $\aval \vinz \g$.
\end{inparaenum}
Then the probabilistic edge $\apedge$ is \emph{piecewise linear} if,
for each $\outcome \in 2^\clocks \times \locs$ and each $\aval \in \valuations$ such that $\aval \vinz \g$,
we have $\adtpara{\pd}{\aval}(\outcome) = \sum_{\aclock \in \clocks} \plf{\apedge,\outcome}{\aclock}(\aval(\aclock))$.
We assume henceforth that all probabilistic edges of cdPTAs are piecewise linear.

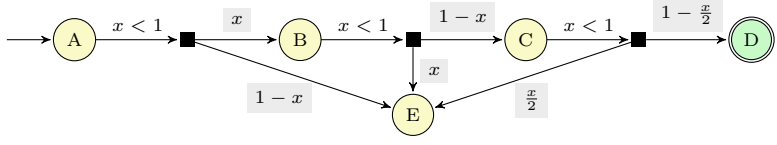
\begin{figure}[t]
{

\centering
\scriptsize

\begin{tikzpicture}[->,>=stealth',shorten >=1pt,auto, thin] 

  \tikzstyle{state}=[draw=black, text=black, shape=circle, inner sep=3pt, outer sep=0pt, circle, rounded corners, fill=ly] 
  \tikzstyle{final_state}=[draw=black, text=black, shape=circle, inner sep=3pt, outer sep=0pt, circle, rounded corners, fill= lg] 

	\node(for_initial){};

	\node[state, node distance=1cm](A)[right of=for_initial]{$\mathrm{A}$};

	\node[fill=black, node distance=1.5cm] (nail_A) [right of=A] {};

	\node[state, node distance=1.5cm](B)[right of=nail_A]{$\mathrm{B}$};

	\node[fill=black, node distance=1.5cm] (nail_B) [right of=B] {};

	\node[state, node distance=1.5cm](C)[right of=nail_B]{$\mathrm{C}$};

	\node[fill=black, node distance=1.5cm] (nail_C) [right of=C] {};

	\node[final_state, node distance=1.5cm,accepting](D)[right of=nail_C]{$\mathrm{D}$};

	\node[state, node distance=1cm](E)[below of=nail_B]{$\mathrm{E}$};
		
	\path[rounded corners]
	(for_initial)
		edge (A)

	(A)
		edge [above] node {$\aclock<1$} (nail_A)

	(nail_A)
		edge [above] node {\colorbox{gray!15}{$\aclock$}} (B)
		edge [below] node[xshift=-2mm] {\colorbox{gray!15}{$1-\aclock$}} (E)

	(B)
		edge [above] node {$\aclock< 1$} (nail_B)

	(nail_B)
		edge [above] node {\colorbox{gray!15}{$1-\aclock$}} (C)
		edge [right] node {\colorbox{gray!15}{$\aclock$}} (E)

	(C)
		edge [above] node {$\aclock<1$} (nail_C)

	(nail_C)
		edge [above] node {\colorbox{gray!15}{$1-\frac{\aclock}{2}$}} (D)
		edge [below] node {\colorbox{gray!15}{$\frac{\aclock}{2}$}} (E);

\end{tikzpicture}

}

\caption{A one-clock cdPTA for which the maximum probability is attained by a time delay corresponding to an irrational number.}
\label{fig:one_clock}
\end{figure}

\paragraph{Example 2.}
%
Standard methods for the analysis of timed automata 
typically consist of a finite-state system that represents faithfully the original model.
In particular, the region graph \cite{AD94} and the corner-point abstraction \cite{BBL08}
both involve the division of the state space
according to a fixed, rational-numbered granularity.
The example of a one-clock cdPTA $\acdpta$ of  Figure~\ref{fig:one_clock}
shows that such an approach
cannot be used for the exact computation of optimal reachability probabilities in cdPTAs,
because optimality may be attained when the clock has an irrational value.
For an example of the formal description of a piecewise linear probabilistic edge,
consider the probabilistic edge from location $\mathrm{C}$,
which we denote by $\apedge_\mathrm{C}$:
then we have $\apartition{\apedge_\mathrm{C},(\emptyset,\mathrm{D})}{\aclock} 
= \apartition{\apedge_\mathrm{C},(\emptyset,\mathrm{E})}{\aclock}
= \{ [0,1) \}$,
with $\fc{\apedge_\mathrm{C},(\emptyset,\mathrm{D})}{\aclock}{[0,1)} = 1$,
$\fd{\apedge_\mathrm{C},(\emptyset,\mathrm{D})}{\aclock}{[0,1)} = -\frac{1}{2}$,
$\fc{\apedge_\mathrm{C},(\emptyset,\mathrm{E})}{\aclock}{[0,1)} = 0$,
and
$\fd{\apedge_\mathrm{C},(\emptyset,\mathrm{E})}{\aclock}{[0,1)} = \frac{1}{2}$.
Now consider the maximum probability of reaching location $\mathrm{D}$
(that is, $\maxval{\semPTS{\acdpta}}{\states_{\{\mathrm{D}\}}}{\cdptastrategies}$).
Intuitively, the longer the cdPTA remains in location $\mathrm{A}$,
the lower the probability of making a transition to location $\mathrm{E}$ from $\mathrm{A}$,
but the higher the probability of making a transition to $\mathrm{E}$ 
from $\mathrm{B}$ and $\mathrm{C}$.
Note that, after $\mathrm{A}$ is left,
the choice resulting in the maximum probability of reaching $\mathrm{D}$
is to take the outgoing transitions from $\mathrm{B}$ and $\mathrm{C}$ as soon as possible
(delaying in $\mathrm{B}$ and $\mathrm{C}$ will increase the value of $\aclock$,
therefore increasing the probability of making a transition to $\mathrm{E}$).
Denoting by $\adelay$ the amount of time elapsed in $\mathrm{A}$,
the maximum probability of reaching $\mathrm{D}$
is equal to $\adelay(1-\adelay)(1-\frac{\adelay}{2})$, 
which (within the interval $[0,1)$) reaches its maximum at $1 - \frac{\sqrt{3}}{3}$.
Hence, this example indicates that abstractions based on the optimality of choices made
at (or arbitrarily close to) rational-numbered clock values 
(such as the region graph or corner-point abstraction)
do not yield exact analysis methods for cdPTAs.
\qed



%% file: undec.tex

\begin{theorem}\label{thm:undec}
The maximal reachability problem is undecidable for cdPTAs
with at least 3 clocks.
\end{theorem}
\begin{proof}[sketch]
We proceed by reducing the non-halting problem for two-counter machines
to the maximal reachability problem for cdPTAs.
The reduction has close similarities to a reduction presented in \cite{ABKMT16}.

A two-counter machine $\twocm = (\instrs,\counters)$
comprises a set $\instrs = \{ \anin_1, ..., \anin_n \}$ of instructions
and a set $\counters = \{ \acount_1, \acount_2 \}$ of counters.
The instructions are of the following form (for $1 \leq i,j,k \leq n$ and $l \in \{1,2\}$):
\begin{enumerate}
\item
$\anin_i : \acount_l:=\acount_l + 1; \mbox{ goto } \anin_j$ (increment $\acount_l$);
\item
$\anin_i : \acount_l:=\acount_l - 1; \mbox{ goto } \anin_j$ (decrement $\acount_l$);
\item
$\anin_i : \mbox{ if } (\acount_l>0) \mbox{ them goto } \anin_j \mbox{ else goto } \anin_k$ (zero check $\acount_l$);
\item
$\anin_n : \mbox{HALT}$ (halting instruction).
\end{enumerate}
A configuration $(\anin,\avalc_1,\avalc_2)$ of a two-counter machine comprises an
instruction $\anin$ and values $\avalc_1$ and $\avalc_2$ of counters $\acount_1$ and $\acount_2$, respectively.
A run of a two-counter machine consists of a finite or infinite sequence of configurations,
starting from configuration $(\anin_1,0,0)$,
and where subsequent configurations are successively generated by following the rule specified in the associated configuration.
A run is finite if and only if the final instruction visited along the run is $\anin_n$ (the halting instruction).
The halting problem for two-counter machines concerns determining whether 
the unique run of the two-counter machine is finite, 
and is undecidable \cite{Min67};
hence the non-halting problem (determining whether 
the unique run of the two-counter machine is infinite) is also undecidable.

Consider a two-counter machine $\twocm$.
We reduce the non-halting problem for $\twocm$ to the maximal reachability problem in the following way.
We construct a cdPTA $\acdpta_\twocm$ with three clocks $\{ \aclock_1, \aclock_2, \aclock_3 \}$ 
by considering modules for each form that the instructions of a two-counter machine can take.
On entry to each module, we have that $\aclock_1 = \frac{1}{2^{\acount_1}}$, $\aclock_2 = \frac{1}{2^{\acount_2}}$ and
$\aclock_3 = 0$.
The module for simulating an increment instruction is shown in Figure~\ref{fig:increment}.
In location $\anin_i$, there is a delay of $1-\frac{1}{2^{\acount_1}}$,
and hence the values of the clocks on entry to location $\mathrm{B}$ are
$\aclock_1 = 0$, $\aclock_2 = \frac{1}{2^{\acount_2}} + 1-\frac{1}{2^{\acount_1}} \mod 1$
and $\aclock_3 = 1-\frac{1}{2^{\acount_1}}$.
A nondeterministic choice is then made concerning the amount of time that elapses in location $\mathrm{B}$:
note that this amount must be in the interval $(0,\frac{1}{2^{\acount_1}})$.
In order to correctly simulate the increment of counter $\acount_1$,
the choice of delay in location $\mathrm{B}$ should be equal to $\frac{1}{2^{\acount_1 + 1}}$.
On leaving location $\mathrm{B}$, a probabilistic choice is made:
the rightward outcome corresponds to continuing the simulation of the two-counter machine,
whereas the downward outcome corresponds to checking that the delay in location $\mathrm{B}$
was correctly $\frac{1}{2^{\acount_1 + 1}}$.
We write the delay in location $\mathrm{B}$ as $\frac{1}{2^{\acount_1 + 1}} + \epsilon$,
where $-\frac{1}{2^{\acount_1 + 1}} < \epsilon < \frac{1}{2^{\acount_1 + 1}}$:
hence, for a correct simulation of the increment of $\acount_1$,
we require that $\epsilon=0$.

Consider the case in which the downward outcome (from the outgoing probabilistic edge of location $\mathrm{B}$)
is taken: then the cdPTA fragment from location $\mathrm{D}$ has the role of checking whether $\epsilon=0$.
Note that, after entering location $\mathrm{D}$, 
no time elapses in locations $\mathrm{D}$ and $\mathrm{E}$
(as enforced by the reset of $\aclock_2$ to zero and the invariant condition $\aclock_2 = 0$),
and hence both clocks $\aclock_1$ and $\aclock_3$ retain the same values that they had when location $\mathrm{B}$ was left.
We show that the probability of reaching the target location $\mathrm{G}$
from location $\mathrm{D}$ is $\frac{1}{4} - \epsilon^2$,
and hence equal to $\frac{1}{4}$ if and only if $\epsilon=0$.
To see that the probability of reaching $\mathrm{G}$ from $\mathrm{D}$ is $\frac{1}{4} - \epsilon^2$,
observe that the probability is equal to 
$\frac{1}{2}(\aclock_1+\aclock_3) 
= \frac{1}{2}(\frac{1}{2^{\acount_1 + 1}} + \epsilon+ (1-\frac{1}{2^{\acount_1 + 1}}) + \epsilon) 
= \frac{1}{2} + \epsilon$
multiplied by $1-\frac{1}{2}(\aclock_1+\aclock_3) 
= \frac{1}{2} - \epsilon$,
i.e., equal to $\frac{1}{4} - \epsilon^2$.
Hence the probability of reaching location $\mathrm{G}$ from location $\mathrm{D}$
is equal to $\frac{1}{4}$ if and only if $\epsilon=0$ (otherwise, the probability is less than $\frac{1}{4}$).

\begin{figure}[t]
{

\centering
\scriptsize

\begin{tikzpicture}[->,>=stealth',shorten >=1pt,auto, thin] 

  \tikzstyle{state}=[draw=black, text=black, shape=circle, inner sep=3pt, outer sep=0pt, circle, rounded corners, fill=ly] 
  \tikzstyle{final_state}=[draw=black, text=black, shape=circle, inner sep=3pt, outer sep=0pt, circle, rounded corners, fill= lg] 

	\node(for_initial){$\aclock_1 = \frac{1}{2^{\acount_1}}$};

	\node[state, node distance=1.5cm](instruction_i)[right of=for_initial]{$\anin_i$};
	\node[yshift=-0.5cm,above of = instruction_i](inv_instruction_i){$\aclock_1,\aclock_2 \leq 1$};

	\node[state, node distance=2.5cm](B)[right of=instruction_i]{$\mathrm{B}$};
	\node[yshift=-0.3cm,above of = B](inv_B){\begin{tabular}{c}{$\aclock_2 \leq 1, $}\\{$\aclock_1, \aclock_3 < 1$}\end{tabular}};

	\node[fill=black, node distance=2.5cm] (nail_1) [right of=B] {};

	\node[state, node distance=2.5cm](C)[right of=nail_1]{$\mathrm{C}$};
	\node[yshift=-0.5cm,above of = C](inv_C){$\aclock_2,\aclock_3 \leq 1$};

	\node[state, node distance=2.5cm](instruction_j)[right of=C]{$\anin_j$};

	\node[state, node distance=1.2cm](D)[below of=nail_1]{$\mathrm{D}$};
	\node[xshift=-0.2cm,right of = D](inv_D){$\aclock_2=0$};
		
	\node[fill=black, node distance=1.2cm] (nail_2) [below of=D] {};

	\node[state, node distance=2.5cm](E)[right of=nail_2]{$\mathrm{E}$};
	\node[yshift=-0.5cm,above of = E](inv_E){$\aclock_2=0$};

	\node[state, node distance=2.5cm](F)[left of=nail_2]{$\mathrm{F}$};

	\node[fill=black, node distance=2.5cm] (nail_3) [right of=E] {};

	\node[final_state, node distance=1.2cm,accepting](G)[below of=nail_3]{$\mathrm{G}$};
	\node[state, node distance=1.2cm](H)[above of=nail_3]{$\mathrm{H}$};

	\path
	(for_initial)
		edge (instruction_i)

	(instruction_i)
		edge [above] node {$\aclock_1=1$} (B)
		edge [below] node {$\{\aclock_1\}$} (B)
		edge [loop below] node {$\aclock_2=1,\{ \aclock_2 \}$} (instruction_i)

	(B)
		edge [above] node {$0 < \aclock_1, \aclock_3 < 1$} (nail_1)
		edge [loop below] node {$\aclock_2=1,\{ \aclock_2 \}$} (B)

	(nail_1)
		edge [above] node {\colorbox{gray!15}{$\frac{1}{2}$}} (C)
		edge [below] node {$\{\aclock_1\}$} (C)
		edge [left] node {\colorbox{gray!15}{$\frac{1}{2}$}} (D)
		edge [right] node {$\{ \aclock_2 \}$} (D)

	(C)
		edge [above] node {$\aclock_3=1$} (instruction_j)
		edge [below] node {$\{\aclock_3\}$} (instruction_j)
		edge [loop below] node {$\aclock_2=1,\{ \aclock_2 \}$} (B)

	(D)
		edge [left] node {$\aclock_2 = 0$} (nail_2)

	(nail_2)
		edge [below] node {\colorbox{gray!15}{$\frac{1}{2}(\aclock_1+\aclock_3)$}} (E)
		edge [below] node {\colorbox{gray!15}{$1-\frac{1}{2}(\aclock_1+\aclock_3)$}} (F)

	(E)
		edge [above] node {$\aclock_2 = 0$} (nail_3)

	(nail_3)
		edge [left] node[pos=0.7] {\colorbox{gray!15}{$1-\frac{1}{2}(\aclock_1+\aclock_3)$}} (G)
		edge [left] node[pos=0.7] {\colorbox{gray!15}{$\frac{1}{2}(\aclock_1+\aclock_3)$}} (H);

\end{tikzpicture}

}

\caption{The cdPTA module for simulating an increment instruction for counter $\acount_1$.}
\label{fig:increment}
\end{figure}
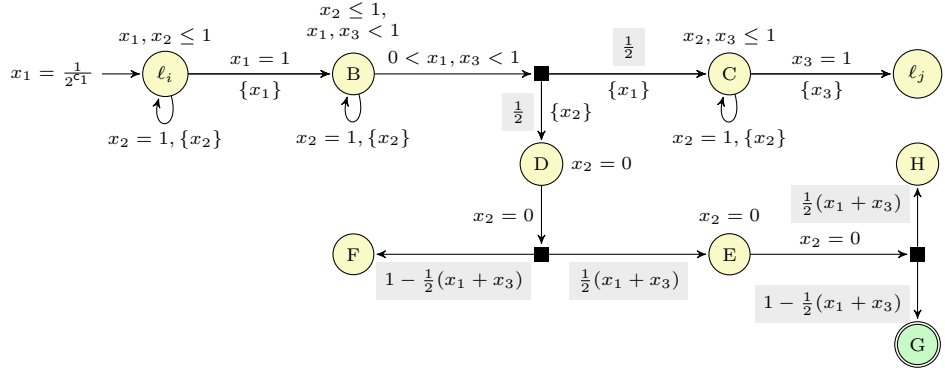


The module for simulating a decrement instruction is shown in Figure~\ref{fig:decrement}.
In a similar manner to the cdPTA fragment in Figure \ref{fig:increment} for the simulation of an increment instruction,
the only nondeterministic choice made is with regard to the amount of time spent in location $\anin_i$,
which is denoted by $\adelay$.
For the correct simulation of the decrement instruction,
$\adelay$ should equal $1 - \frac{1}{2^{\acount_1-1}}$.
The rightward outcome is taken from the probabilistic edge leaving location $\anin_i$
corresponds to the continuation of the simulation of the two-counter machine:
hence, on entry to location $\mathrm{B}$, 
we have $\aclock_1=0$, $\aclock_2=\frac{1}{2^{\acount_2}} + \adelay$ and $\aclock_3=\adelay$;
then, on entry to location $\anin_j$,
we have $\aclock_1=\adelay$, $\aclock_2=\frac{1}{2^{\acount_2}}$ and $\aclock_3=0$.

Let $\adelay = 1 - \frac{1}{2^{\acount_1-1}} + \epsilon$.
For the correct simulation of the decrement instruction,
we require that $\epsilon=0$.
The downward outcome from the probabilistic edge leaving location $\anin_i$
corresponds to checking that $\epsilon=0$,
and takes a similar form to the analogous downward edge of the cdPTA fragment for the increment instruction,
as shown in Figure \ref{fig:increment}.
Note that, on entry to location $\mathrm{C}$,
we have that $\aclock_1=1 - \frac{1}{2^{\acount_1}} + \epsilon$, $\aclock_2=0$ 
and $\aclock_3=1 - \frac{1}{2^{\acount_1-1}} + \epsilon$.
Then, on entry to location $\mathrm{D}$,
we have that $\aclock_1=0$, $\aclock_2=\frac{1}{2^{\acount_1}} - \epsilon$ 
and $\aclock_3=1 - \frac{1}{2^{\acount_1}}$.
As no time elapses in locations $\mathrm{D}$ and $\mathrm{E}$,
we have that target location $\mathrm{F}$ is then reached with probability 
$\frac{1}{2}(\aclock_2+\aclock_3) 
= \frac{1}{2}(\frac{1}{2^{\acount_1}} - \epsilon + 1 - \frac{1}{2^{\acount_1}})
= \frac{1}{2} + \frac{\epsilon}{2}$
multiplied by the probability
$1-\frac{1}{2}(\aclock_2+\aclock_3)
= \frac{1}{2} - \frac{\epsilon}{2}$,
which equals $\frac{1}{4}-\frac{\epsilon^2}{4}$.
Hence we conclude that the probability of reaching location $\mathrm{F}$ from location $\mathrm{C}$
is equal to $\frac{1}{4}$ if and only if $\epsilon=0$.

\begin{figure}[t]
{

\centering
\scriptsize

\begin{tikzpicture}[->,>=stealth',shorten >=1pt,auto, thin] 

  \tikzstyle{state}=[draw=black, text=black, shape=circle, inner sep=3pt, outer sep=0pt, circle, rounded corners, fill=ly] 
  \tikzstyle{final_state}=[draw=black, text=black, shape=circle, inner sep=3pt, outer sep=0pt, circle, rounded corners, fill= lg] 


	\node(for_initial){\begin{tabular}{c}{$\aclock_1 = \frac{1}{2^{\acount_1}},$}\\{$\aclock_3=0$}\end{tabular}};

	\node[state, node distance=1.5cm](instruction_i)[right of=for_initial]{$\anin_i$};
	\node[yshift=0.5cm,below of = instruction_i](inv_instruction_i){$\aclock_1,\aclock_3 < 1$};

	\node[fill=black, node distance=2cm] (nail_1) [right of=A] {};

	\node[state, node distance=2cm](B)[right of=nail_1]{$\mathrm{B}$};
	\node[yshift=0.5cm,below of = B](inv_B){$\aclock_2 \leq 1$};

	\node[state, node distance=2cm](instruction_j)[right of=B]{$\anin_j$};

	\node[state, node distance=2cm](C)[below of=nail_1]{$\mathrm{C}$};
	\node[yshift=0.5cm,below of = C](inv_C){$\aclock_1 \leq 1$};

	\node[state, node distance=2cm](D)[right of=C]{$\mathrm{D}$};
	\node[yshift=0.5cm,below of = D](inv_D){$\aclock_1=0$};
		
	\node[fill=black, node distance=2cm] (nail_2) [right of=D] {};

	\node[state, node distance=2cm](E)[right of=nail_2]{$\mathrm{E}$};
	\node[yshift=0.5cm,below of = E](inv_E){$\aclock_1=0$};

	\node[state, node distance=1.5cm](G)[below of=nail_2]{$\mathrm{G}$};

	\node[fill=black, node distance=2cm] (nail_3) [right of=E] {};

	\node[final_state, node distance=1.5cm,accepting](F)[above of=nail_3]{$\mathrm{F}$};
	\node[state, node distance=1.5cm](H)[below of=nail_3]{$\mathrm{H}$};


	\path
	(for_initial)
		edge (instruction_i)

	(instruction_i)
		edge [above] node {$\aclock_1=1$} (nail_1)

	(nail_1)
		edge [above] node {\colorbox{gray!15}{$\frac{1}{2}$}} (B)
		edge [below] node {$\{\aclock_1\}$} (B)
		edge [left] node {\colorbox{gray!15}{$\frac{1}{2}$}} (C)
		edge [right] node {$\{ \aclock_2 \}$} (C)

	(B)
		edge [above] node {$\aclock_3 = 1$} (instruction_j)
		edge [below] node {$\{\aclock_3\}$} (instruction_j)
		edge [loop above] node {$\aclock_2=1,\{ \aclock_2 \}$} (B)

	(C)
		edge [above] node {$\aclock_1=1$} (D)
		edge [below] node {$\{\aclock_1\}$} (D)

	(D)
		edge [above] node {$\aclock_1 = 0$} (nail_2)

	(nail_2)
		edge [above] node {\colorbox{gray!15}{$\frac{1}{2}(\aclock_2+\aclock_3)$}} (E)
		edge [left] node[pos=0.75] {\colorbox{gray!15}{$1-\frac{1}{2}(\aclock_2+\aclock_3)$}} (G)

	(E)
		edge [above] node {$\aclock_1 = 0$} (nail_3)

	(nail_3)
		edge [left] node[pos=0.75] {\colorbox{gray!15}{$1-\frac{1}{2}(\aclock_2+\aclock_3)$}} (F)
		edge [left] node[pos=0.75] {\colorbox{gray!15}{$\frac{1}{2}(\aclock_2+\aclock_3)$}} (H);


\end{tikzpicture}

}

\caption{The cdPTA module for simulating a decrement instruction for counter $\acount_1$.}
\label{fig:decrement}
\end{figure}
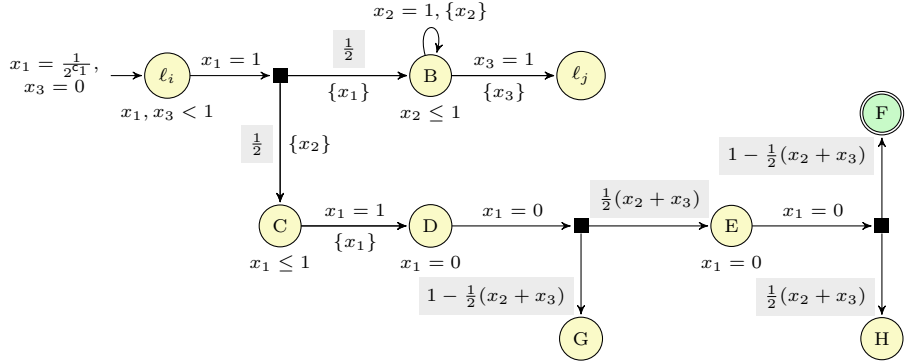

Finally, the module for a zero test instruction 
$\anin_i : \mbox{ if } (\acount_1>0) \mbox{ then goto } \anin_j \mbox{ else goto } \anin_k$
is shown in Figure~\ref{fig:zerotest}.
The module is almost identical to that of \cite{ABKMT16arxiv}, and we present it here only for completeness.
After entry to location $\anin_i$,
two probabilistic edges are enabled:
the rightward one is taken if $\acount_1 = 0$ (i.e., if $\aclock_1 = \frac{1}{2^0} = 1$),
whereas the leftward one is taken otherwise. 
Both probabilistic edges involve an outcome leading to a target location with probability $\frac{1}{4}$:
if this outcome is not taken, the cdPTA fragment then proceeds to location $\anin_j$ or $\anin_j$,
depending on which probabilistic edge was taken.

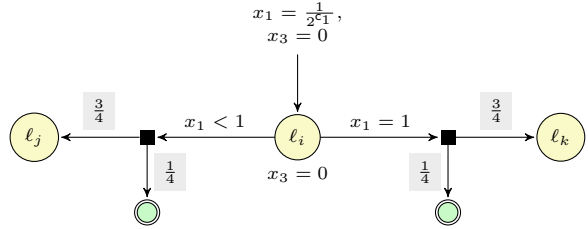
\begin{figure}[t]
{

\centering
\scriptsize

\begin{tikzpicture}[->,>=stealth',shorten >=1pt,auto, thin] 

  \tikzstyle{state}=[draw=black, text=black, shape=circle, inner sep=3pt, outer sep=0pt, circle, rounded corners, fill=ly] 
  \tikzstyle{final_state}=[draw=black, text=black, shape=circle, inner sep=3pt, outer sep=0pt, circle, rounded corners, fill= lg] 


	\node(for_initial){\begin{tabular}{c}{$\aclock_1 = \frac{1}{2^{\acount_1}},$}\\{$\aclock_3=0$}\end{tabular}};

	\node[state, node distance=1.5cm](instruction_i)[below of=for_initial]{$\anin_i$};
	\node[yshift=0.5cm,below of = instruction_i](inv_instruction_i){$\aclock_3 = 0$};

	\node[fill=black, node distance=2cm] (nail_1) [right of=instruction_i] {~};
	\node[fill=black, node distance=2cm] (nail_2) [left of=instruction_i] {~};

	\node[final_state, node distance=1cm,accepting](top)[below of=nail_1]{};
	\node[state, node distance=1.5cm](instruction_k)[right of=nail_1]{$\anin_k$};

	\node[final_state, node distance=1cm,accepting](bottom)[below of=nail_2]{};
	\node[state, node distance=1.5cm](instruction_j)[left of=nail_2]{$\anin_j$};


	\path
	(for_initial)
		edge (instruction_i)

	(instruction_i)
		edge [above] node {$\aclock_1=1$} (nail_1)
		edge [above] node {$\aclock_1<1$} (nail_2)

	(nail_1)
		edge [left] node {\colorbox{gray!15}{$\frac{1}{4}$}} (top)
		edge [above] node {\colorbox{gray!15}{$\frac{3}{4}$}} (instruction_k)

	(nail_2)
		edge [right] node {\colorbox{gray!15}{$\frac{1}{4}$}} (bottom)
		edge [above] node {\colorbox{gray!15}{$\frac{3}{4}$}} (instruction_j);


\end{tikzpicture}

}

\caption{The cdPTA module for simulating a zero-test instruction for counter $\acount_1$.}
\label{fig:zerotest}
\end{figure}

Given the construction of a cdPTA simulating the two-counter machine using the modules described above,
we can now proceed to show Theorem~\ref{thm:undec}.
The reasoning is the same as that of Lemma~5 of \cite{ABKMT16}.
If the two-counter machine halts in $k$ steps,
and the strategy of the cdPTA correctly simulates the two-counter machine
the probability of reaching a target location will be 
$\frac{1}{2} \cdot \frac{1}{4} + (\frac{1}{2})^2 \cdot \frac{1}{4} + ... + (\frac{1}{2})^k \cdot \frac{1}{4} < \frac{1}{4}$.
If the two-counter machine halts in $k$ steps,
and the strategy of the cdPTA does not correctly simulate the two-counter machine,
then this means that the probability of reaching a target location is strictly less than that corresponding to correct simulation,
given that deviation from simulation of a certain step corresponds to 
reaching the target locations with probability strictly less than $\frac{1}{4}$ in that step.
Now consider the case in which the two-counter machine does not halt:
in this case, faithful simulation in the cdPTA corresponds to reaching target locations 
with probability $\sum_{i=1}^\infty (\frac{1}{2})^i \cdot \frac{1}{4} = \frac{1}{4}$,
whereas unfaithful simulation in the cdPTA corresponds to reaching the target locations with probability 
$\sum_{i=1}^\infty (\frac{1}{2})^i \cdot \gamma_i$
where $\gamma_i \leq \frac{1}{4}$ for all $i \in \Nset$
and $\gamma_j < \frac{1}{4}$ for at least one $j \in \Nset$,
and hence $\sum_{i=1}^\infty (\frac{1}{2})^i \cdot \gamma_i < \frac{1}{4}$.
Therefore the two-counter machine does not halt 
if and only if there exists a strategy in the constructed cdPTA that reaches the target locations
with probability at least $\frac{1}{4}$,
concluding the proof of Theorem~\ref{thm:undec}.
\qed
\end{proof}

%% file: approx.tex
We now consider the approximation of maximal and minimal reachability probabilities of cdPTAs.
Our approach is to utilise concepts from the corner-point abstraction \cite{BBL08}.
However, while the standard corner-point abstraction is a finite-state system that 
extends the classical region graph by encoding corner points within states,
the states of our finite-state system correspond to regions, 
and we use corners of regions only to define available distributions.
Furthermore, in contrast to the widespread use of the corner-point abstraction in the context
of weighted (or priced) timed automata (see \cite{Bou15} for a survey),
and in line with the undecidability results presented in Section~\ref{sec:undec},
our variant of the corner-point abstraction does not result in a finite-state system that can be used
to obtain a quantitative measure that is arbitrarily close to the actual one:
in the context of cdPTAs, we will present a method that approximates maximal and minimal reachability properties,
and show that successive refinement of regions leads to a more accurate approximation.

First we define regions and corner points.
Let $\acdpta = (\locs, \locinit, \clocks, \inv, \pedges)$ be a cdPTA,
which we assume to be fixed throughout this section,
and let $\maxbound \in \Nset$ denote the upper bound on clocks in $\acdpta$.
We choose $k \in \Nset$, which we will refer to as the \emph{(time) granularity},
and let $\multof{k} = \{ \frac{c}{k} \setsep c \in \Nset \}$ be the set of multiples of $\frac{1}{k}$.
A \emph{$k$-region} $(\anatval,[\aclockset_0, ..., \aclockset_n])$ over $\clocks$ 
comprises: 
\begin{enumerate}
\item
a function $\anatval: \clocks \ra (\multof{k} \cap [0,\maxbound])$ assigning a multiple of $\frac{1}{k}$ 
no greater than $\maxbound$ to each clock
and 
\item
a partition $[\aclockset_0, ..., \aclockset_n]$ of $\clocks$,
where $\aclockset_i \neq \emptyset$ for all $1 \leq i \leq n$ 
and $\anatval(\aclock)=\maxbound$ implies $\aclock \in \aclockset_0$ for all $\aclock \in \clocks$.
\end{enumerate}
%
%
%
Given clock valuation $\aval \in \valuations$ and granularity $k$,
the \emph{$k$-region $\areg = (\anatval,[\aclockset_0, ..., \aclockset_n])$ 
containing $\aval$} (written $\aval \in \areg$)
satisfies the following conditions:
\begin{enumerate}
\item
$\floor{k \cdot \aval(\aclock)} {=} k \cdot \anatval(\aclock)$ for all clocks $\aclock \in \clocks$;
\item
$\aval(\aclock) {=} \anatval(\aclock)$ for all clocks $\aclock \in \aclockset_0$;
\item
$k \cdot \aval(\aclock) - \floor{k \cdot \aval(\aclock)} \leq k \cdot \aval(\anotherclock) - \floor{k \cdot \aval(\anotherclock)}$
if and only if $\aclock \in \aclockset_i$ and $\anotherclock \in \aclockset_j$ with $i \leq j$,
for all clocks $x,y \in \clocks$.
\end{enumerate}
%
Note that, rather than considering regions delimited by valuations corresponding to natural numbers,
in our definition regions are delimited by valuations corresponding to multiples of $\frac{1}{k}$. 
%
%
We use $\regs{\maxbound}{k}$ to denote the set of $k$-regions. 
For $\areg, \areg' \in \regs{\maxbound}{k}$ and clock constraint $\cc \in \ccons$,
we say that $\areg'$ is a \emph{$\cc$-satisfying time successor} of $\areg$
if there exist $\aval \in \areg$ and $\adelay \in \nnr$
such that $(\aval{+}\adelay) \in \areg'$ and $(\aval{+}\adelay') \vinz \cc$ for all $0 \leq \adelay' \leq \adelay$.
For a given $k$-region $\areg \in \regs{\maxbound}{k}$,
we let $\areg[\aclockset:=0]$ be the $k$-region that corresponds to 
resetting clocks in $\aclockset$ to 0 from clock valuations in $\areg$
(that is, $\areg[\aclockset:=0]$ contains valuations $\aval[\aclockset:=0]$ for $\aval \in \areg$).
We use $\areg_\bzero$ to denote the $k$-region that contains the valuation $\bzero$.

A \emph{corner point} $\acp = \seq{\acpconst_i}{0 \leq i \leq n} \in (\multof{k} \cap [0,\maxbound])^n$
of $k$-region $(\anatval,[\aclockset_0, ..., \aclockset_n])$ 
is defined by:
\[
\acpconst_i(\aclock) =
\left\{
\begin{array}{ll}
\anatval(\aclock) 
&
\mbox{if } \aclock \in \aclockset_j \mbox{ with } j \leq i
\\
\anatval(\aclock) + \frac{1}{k}
&
\mbox{if } \aclock \in \aclockset_j \mbox{ with } j > i \; .
\end{array}
\right.
\]
Note that a $k$-region $(\anatval,[\aclockset_0, ..., \aclockset_n])$ is associated with $n + 1$ corner points.
Let $\cpsofreg{\areg}$ be the set of corner points of $k$-region $\areg$.
Given 
granularity $k$,
we let $\cps{\maxbound}{k}$ be the set of all corner points.

Next we define the \emph{clock-dependent region graph with granularity $k$}
as the finite-state PTS $\acdrg{\acdpta}{k} = ( \rgstates{\acdpta}{k}, \rginit, \rgactions{\acdpta}{k}, \rgtrans{\acdpta}{k} )$,
where $\rgstates{\acdpta}{k} = \locs \times \regs{\maxbound}{k}$,
$\rginit = (\locinit, \areg_\bzero)$,
$\rgactions{\acdpta}{k} = \{ \rsuccact \} \cup (\cps{\maxbound}{k} \times \pedges)$,
and $\rgtrans{\acdpta}{k} = \rgtransdelay{\acdpta}{k} \cup \rgtranspedge{\acdpta}{k}$ 
where
$\rgtransdelay{\acdpta}{k} \subseteq 
\rgstates{\acdpta}{k} \times \{ \rsuccact \} \times \dist(\rgstates{\acdpta}{k})$
and $\rgtranspedge{\acdpta}{k} \subseteq 
\rgstates{\acdpta}{k} \times \cps{\maxbound}{k} \times \pedges \times \dist(\rgstates{\acdpta}{k})$
such that:
\begin{itemize}
\item
$\rgtransdelay{\acdpta}{k}$ is the smallest set of transitions 
such that $((\aloc,\areg),\rsuccact,\dirac{(\aloc,\areg')}) \in \rgtransdelay{\acdpta}{k}$
if $(\aloc,\areg')$ is an $\inv(\aloc)$-satisfying time successor of $(\aloc,\areg)$; 
\item
$\rgtranspedge{\acdpta}{k}$ is the smallest set 
such that $((\aloc,\areg),(\acp,(\aloc,\g,\pd)),\anrgdist) \in \rgtranspedge{\acdpta}{k}$ if:
\begin{enumerate}
\item
$\areg \vinz \g$;
\item
$\acp \in \cpsofreg{\areg}$;
\item
for any $(\aloc',\areg') \in \rgstates{\acdpta}{k}$, we have that
$\anrgdist(\aloc',\areg') = \sum_{\aclockset \in \reset{\areg}{\areg'}} \adtpara{\pd}{\acp}(\aclockset,\aloc')$,
where $\reset{\areg}{\areg'} = \{ \aclockset \subseteq \clocks \mid \areg[\aclockset:=0] = \areg' \}$.
\end{enumerate}
\end{itemize}

Hence the clock-dependent region graph of a cdPTA
encodes corner points within (probabilistic-edge-based) transitions,
in contrast to the corner-point abstraction,
which encodes corner points within states.
In fact, a literal application of the standard corner-point abstraction,
as presented in \cite{Bou15},
does not result in a conservative approximation,
which we now explain with reference to Example~2.
%

\paragraph{Example~2 (continued).}
Recall that the states of the corner-point abstraction comprise 
a location, a region and a corner point of the region,
and transitions maintain consistency between corner points
of the source and target states.
For example, for the cdPTA of Figure~\ref{fig:one_clock}, 
consider the state $(\mathrm{A},0<\aclock<1,\aclock=1)$,
where $0<\aclock<1$ is used to refer to the state's region component 
and $\aclock=1$ is used to refer to the state's corner point.
Then the probabilistic edge leaving location $\mathrm{A}$ is enabled
(because the state represents the situation in which clock $\aclock$ is in the interval $(0,1)$
and arbitrarily close to $1$).
Standard intuition on the corner-point abstraction 
(adapted from weights in \cite{Bou15} to probabilities in distribution templates in this paper)
specifies that, when considering probabilities of outgoing probabilistic edges,
the state $(\mathrm{A},0<\aclock<1,\aclock=1)$
should be associated with probabilities for which $\aclock=1$.
Hence the probability of making a transition to location $\mathrm{B}$ is $1$,
and the target corner-point-abstraction state is $(\mathrm{B},0<\aclock<1,\aclock=1)$.
However, now consider the probabilistic edge leaving location $\mathrm{B}$:
in this case, given that the corner point under consideration is $\aclock=1$,
the probability of making a transition to location $\mathrm{C}$ is $0$,
and hence the target location $\mathrm{D}$ is reachable with probability $0$.
Furthermore, consider the state $(\mathrm{A},0<\aclock<1,\aclock=0)$:
in this case, if the probabilistic edge leaving location $\mathrm{A}$ is taken,
then location $\mathrm{B}$ is reached with probability $0$,
and hence location $\mathrm{D}$ is again reachable with probability $0$.
We can conclude that such a direct application of the corner-point abstraction to cdPTA
is not a conservative approximation of the cdPTA,
because the maximum reachability probability in the corner-point abstraction is 0,
i.e., less than the 
maximum reachability probability of the cdPTA (which we recall is $1-\frac{\sqrt{3}}{3}$).
Instead, in our definition of the clock-dependent region graph,
we allow ``inconsistent'' corner points to be used in successive transitions:
for example, from location $\mathrm{A}$,
the outgoing probabilistic edge can be taken using the value of $\aclock$ 
corresponding to the corner point $\aclock=1$;
then, from locations $\mathrm{B}$ and $\mathrm{C}$,
the outgoing probabilistic edge can be taken using corner point $\aclock=0$.
Hence maximum probability of reaching the target location $\mathrm{D}$, with $k=1$, is 1.
\qed

Analogously to the case of cdPTA strategies,
we consider strategies of clock-dependent region graphs 
that alternate between transitions from $\rgtransdelay{\acdpta}{k}$
(time elapse transitions)
and transitions from $\rgtranspedge{\acdpta}{k}$
(probabilistic edge transitions).
Formally,
a \emph{region graph strategy} $\astrat$ is a strategy of $\acdrg{\acdpta}{k}$ such that, 
for a finite run $\afinrun \in \finruns{\acdrg{\acdpta}{k}}$ 
that has $(\aloc,\areg) \xrightarrow{\anaction,\anrgdist} (\aloc',\areg')$ as its final transition,
either $((\aloc,\areg),\anaction,\anrgdist) \in \rgtransdelay{\acdpta}{k}$ 
and $\support(\astrat(\afinrun)) \in \rgtranspedge{\acdpta}{k}$,
or $((\aloc,\areg),\anaction,\anrgdist) \in \rgtranspedge{\acdpta}{k}$ 
and $\support(\astrat(\afinrun)) \in \rgtransdelay{\acdpta}{k}$.
We write $\rgstrategies{\acdpta}{k}$ for the set of region graph strategies of $\acdrg{\acdpta}{k}$.

Let $\finallocs \subseteq \locs$ be the set of target locations,
which we assume to be fixed in the following.
Recall that $\finalstates = \{ (\aloc,\aval) \in \locs \times \valuations \setsep \aloc \in \finallocs \}$
and let 
$\finalregs{\maxbound}{k} = \{ (\aloc,\areg) \in \rgstates{\maxbound}{k} \setsep \aloc \in \finallocs \}$.
The following result specifies that the maximum (minimum) probability for reaching target locations from the initial state
of a cdPTA is bounded from above (from below, respectively) by the corresponding 
maximum (minimum, respectively) probability in the clock-dependent region graph with granularity $k$.
Similarly, the maximum (minimum) probability computed in the region graph of granularity $k$
is an upper (lower, respectively) bound on the maximum (minimum, respectively) probability 
computed in the region graph of granularity $2k$
(we note that this result can be adapted 
to hold for granularity $ck$ rather than $2k$, for  any $c \in \Nset \setminus \{ 0, 1 \}$).
The proof of the proposition can be found in the appendix.

\begin{proposition}\label{prop:approx}
~
\begin{enumerate}
\item
$\maxval{\semPTS{\acdpta}}{\finalstates}{\cdptastrategies}
\leq
\maxval{\acdrg{\acdpta}{k}}{\finalregs{\maxbound}{k}}{\rgstrategies{\acdpta}{k}}$,
$\minval{\semPTS{\acdpta}}{\finalstates}{\cdptastrategies} \geq 
\minval{\acdrg{\acdpta}{k}}{\finalregs{\maxbound}{k}}{\rgstrategies{\acdpta}{k}}$.
\item
$\maxval{\acdrg{\acdpta}{2k}}{\finalregs{\maxbound}{2k}}{\rgstrategies{\acdpta}{2k}}
\leq 
\maxval{\acdrg{\acdpta}{k}}{\finalregs{\maxbound}{k}}{\rgstrategies{\acdpta}{k}}$,
$\minval{\acdrg{\acdpta}{2k}}{\finalregs{\maxbound}{2k}}{\rgstrategies{\acdpta}{2k}}
\geq 
\minval{\acdrg{\acdpta}{k}}{\finalregs{\maxbound}{k}}{\rgstrategies{\acdpta}{k}}$.
\end{enumerate}
\end{proposition}

\paragraph{Example~2 (continued).}
We give the intuition underlying Proposition~\ref{prop:approx} using Example~2 (Figure~\ref{fig:one_clock}),
considering the maximum probability of reaching the target location $\mathrm{D}$.
When $k=1$, as described above, 
the maximum probability of reaching $\mathrm{D}$ is $1$.
Instead, for $k=2$, the maximum probability of reaching location $\mathrm{D}$ corresponds to taking
the probabilistic edge from location $\mathrm{A}$ for the corner point $\aclock = \frac{1}{2}$
corresponding to the $2$-region $0<\aclock<\frac{1}{2}$
and the probabilistic edges from locations $\mathrm{B}$ and $\mathrm{C}$ for corner point $\aclock = 0$,
again for the $2$-region $0<\aclock<\frac{1}{2}$
i.e., the probability is $\frac{1}{2}$.
With granularity $k=4$, the maximum probability of reaching location $\mathrm{D}$ is $0.328125$,
obtained by taking the probabilistic edge from $\mathrm{A}$ for the corner point $\aclock = \frac{1}{2}$,
and the probabilistic edges from $\mathrm{B}$ and $\mathrm{C}$ for corner point $\aclock = \frac{1}{4}$,
where the $4$-region used in all cases is $\frac{1}{4}<\aclock<\frac{1}{2}$.
\qed



\begin{figure}[t]
\centering
\includegraphics[scale=0.4]{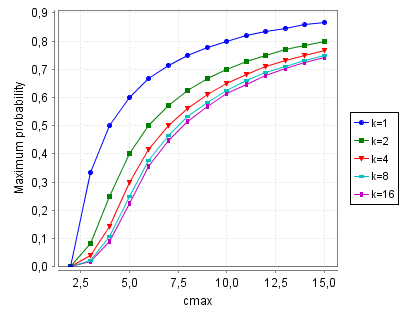}
\caption{Maximum probability of reaching location $\checkmark$ in the cdPTA of Figure~\ref{fig:robot}.
\label{fig:results}}
\end{figure}

\paragraph{Example~1 (continued).}
In Figure~\ref{fig:results} we plot the values of the maximum probability of reaching location $\checkmark$
in the example of Figure~\ref{fig:robot}
for various values of $c_\mathrm{max}$ and $k$,
obtained by encoding the clock-dependent region graph as a finite-state PTS
and using {\sc Prism} \cite{prism}.
For this example, the difference between 
the probabilities obtained from low values of $k$ is substantial.
We note that the number of states of the largest instance that we considered here 
(for $k=16$ and $c_\mathrm{max}=15$) was 140174.
\qed

%% file: concl.tex
In this paper we  presented cdPTAs,
an extension of PTAs in which probabilities can depend on the values of clocks.
We have shown that a basic probabilistic model checking problem,
maximal reachability, is undecidable for cdPTAs with at least three clocks.
One direction of future research could be attempting to improve these results by considering cdPTAs with one or two clocks,
or identifying other kinds of subclass of cdPTAs for which for which probabilistic reachability is decidable:
for example, we conjecture decidability can be obtained 
for cdPTAs in which all clock variables are reset after utilising a probabilistic edge 
that depends non-trivially on clock values.
Furthermore, we conjecture that qualitative reachability problems
(whether there exists a strategy such that the target locations are reached with probability strictly greater than $0$,
or equal to $1$)
are decidable (and in exponential time) for cdPTAs for which the piecewise linear functions are bounded away from $0$
by a region graph construction.
The case of piecewise linear functions that can approach arbitrarily closely to $0$
requires more care
(because non-forgetful cycles, in the terminology of \cite{BA11}, 
can lead to convergence of a probability used along a cdPTA path to $0$).
We also presented a conservative overapproximation method for cdPTAs.
At present this method gives no guarantees on the distance of the obtained bounds to the actual optimal probability:
future work could address this issue, by extending the region graph construction 
from a PTS to a stochastic game 
(to provide upper and lower bounds on the maximum/minimum probability in the manner of \cite{KKNP10}),
or by considering approximate relations (by generalising the results of \cite{DAK12,BA17} 
from Markov chains to PTSs).

\subsubsection{Acknowledgments.}
The inspiration for cdPTA arose from a discussion with Patricia Bouyer on the corner-point abstraction.
Thanks also to Holger Hermanns,
who expressed interest in a cdPTA-like formalism in a talk at Dagstuhl Seminar 14441.

%% file: appendix.tex

\subsection{Preliminaries}
Given set $Q$,
let $\set{\adist_i}{i \in I} \subseteq \dist(Q)$ be a set of distributions
and $\set{\ctweight_i}{i \in I}$ be a set of weights such that $\ctweight_i > 0$ for all $i \in I$
and $\sum_{i \in I} \ctweight_i = 1$.
Then we write $\distsum_{i \in I} \ctweight_i \cdot \adist_i$ to refer to the distribution over $Q$
such that $(\distsum_{i \in I} \ctweight_i \cdot \adist_i)(q) = \sum_{i \in I} \ctweight_i \cdot \adist_i(q)$
for each $q \in Q$.

Let $\anequiv \subseteq \states \times \states$ be an equivalence relation over $\states$.
We say that $\anequiv$ \emph{respects} $\astateset \subseteq \states$
if $\astateset$ is the union of states contained in some set of equivalence classes of $\anequiv$.
Given two distributions $\adist,\adist'$ over $\states$,
we write $\adist \anequiv \adist'$ if 
$\sum_{\astate \in \anec} \adist(\astate) 
= 
\sum_{\astate \in \anec} \adist'(\astate)$
for all equivalence classes $\anec$ of $\anequiv$.
A \emph{combined transition} from state $\astate \in \states$ is a pair 
$(\set{(\astate,\anaction_i,\adist_i)}{i \in I}, \set{\ctweight_i}{i \in I})$
such that $(\astate,\anaction_i,\adist_i) \in \ptstrans$ and $\ctweight_i > 0$
for all $i \in I$, and $\sum_{i \in I} \ctweight_i = 1$.
Let $\actset \subseteq \actions$ be a set of actions.
Then a \emph{probabilistic simulation respecting $\anequiv$ and $\actset$} is a relation
$\aprobsim \subseteq \states \times \states$ such that $\astate \aprobsim \anotherstate$
implies that (1)~$\astate \anequiv \anotherstate$,
and (2)~for each transition $(\astate,\anaction,\adist) \in \ptstrans$,
there exists a combined transition 
$(\set{(\anotherstate,\anaction_i,\adist_i)}{i \in I}, \set{\ctweight_i}{i \in I})$
such that $\adist \anequiv \distsum_{i \in I} \ctweight_i \cdot \adist_i$,
$\set{\anaction_i}{i \in I} \subseteq \actset$ if $\anaction \in \actset$,
and $\set{\anaction_i}{i \in I} \subseteq \actions \setminus \actset$ 
if $\anaction \in \actions \setminus \actset$.\footnote{
Our notion of probabilistic simulation respecting an equivalence relation
is stronger than that of probabilistic simulation of \cite{SegPhD}.
Also note that we do not require actions to be matched in 
the definition of probabilistic simulation respecting $\anequiv$,
although we \emph{do} require that matching actions 
are either all in $\actset$ or all in $\actions \setminus \actset$.}
%

Next, we consider strategies that alternate between actions 
in a certain set $\actset \subseteq \actions$
and actions in the complement set $\actions \setminus \actset$.
Formally, an \emph{$\actset$-alternating strategy} $\astrat$ is a strategy such that,
for finite run $\afinrun \in \finruns{\aPTS}$ that has $\astate \xrightarrow{\anaction,\adist} \astate'$
as its final transition,
then 
$\{ \anaction' \in \actions \setsep (\astate,\anaction',\adist) \in \support(\astrat(\afinrun)) \}
\subseteq \actset$
if $\anaction \in \actions \setminus \actset$,
and $\{ \anaction' \in \actions \setsep (\astate,\anaction',\adist) \in \support(\astrat(\afinrun)) \}
\subseteq \actions \setminus \actset$
if $\anaction \in \actset$.
Let $\altstrategiespts{\aPTS}{\actset}$ be the set of $\actset$-alternating strategies of $\aPTS$;
when the context is clear, we write simply $\altstrategies{\actset}$ 
rather than $\altstrategiespts{\aPTS}{\actset}$.

Given two PTSs $\aPTS_1 = ( \states_1, \sinit_1, \actions_1, \ptstrans_1 )$
and $\aPTS_2 = ( \states_2, \sinit_2, \actions_2, \ptstrans_2 )$,
their disjoint union is defined as the PTS
$( \states_1 \uplus \states_2, \_, \actions_1 \uplus \actions_2, \ptstrans_1 \uplus \ptstrans_2 )$
(where the initial state is irrelevant and is hence omitted).
%
The following result is essentially identical to \cite[Lemma~3.17, Lemma~3.18]{Hah13}
(which in turn rely on \cite[Theorem~8.6.1]{SegPhD}).
%
\begin{proposition}\label{prop:hahn}\cite{Hah13}
Let $\actset_1 \subseteq \actions_1$, let $\actset_2 \subseteq \actions_2$,
and let $\anequiv$ be an equivalence relation over $\states_1 \uplus \states_2$ that respects $\finalstates$ .
If $\sinit_1 \aprobsim \sinit_2$ for a probabilistic simulation respecting $\anequiv$
and $\actset_1 \uplus \actset_2$,
then 
$\maxval{\aPTS_1}{\finalstates}{\altstrategies{\actset_1}} 
\leq 
\maxval{\aPTS_2}{\finalstates}{\altstrategies{\actset_2}}$
and 
$\minval{\aPTS_1}{\finalstates}{\altstrategies{\actset_1}} 
\geq 
\minval{\aPTS_2}{\finalstates}{\altstrategies{\actset_2}}$.
\end{proposition}

\subsection{Approximating a cdPTA by the clock-dependent region graph with granularity $k$}

In order to show part~(1) of Proposition~\ref{prop:approx},
we first consider the following intermediate lemmata.
The first lemma specifies that
the sets of clocks that, when reset to 0, are used to transform valuation $\aval$ to valuation $\aval'$
are the same as the sets of clocks used to transform the $k$-region containing $\aval$
to the $k$-region containing the valuation $\aval'$.

\begin{lemma}\label{lem:resets}
Let $\maxbound, k \in \Nset$ and $\aval, \aval' \in \valuations$
such that, for each clock $\aclock \in \clocks$, either $\aval'(\aclock)=\aval(\aclock)$ or $\aval'(\aclock)=0$.
Using $\areg, \areg' \in \regs{\maxbound}{k}$ to denote the $k$-regions such that $\aval \in \areg$ and $\aval' \in \areg'$,
we have $\reset{\aval}{\aval'} = \reset{\areg}{\areg'}$.
\end{lemma}
\begin{proof}
%
Let $\aclockset^0_\aval$ be the set of clocks that are equal to 0 in $\aval$,
and let $\aclockset^0_{\aval'}$ be the set of clocks that are equal to 0 in $\aval'$.
Similarly, let $\aclockset^0_\areg$ be the set of clocks that are equal to 0 in valuations in $\areg$,
and let $\aclockset^0_{\areg'}$ be the set of clocks that are equal to 0 in valuations in $\areg'$.
By the definition of $k$-regions,
for any clock $\aclock \in \clocks$,
$\aval(\aclock)=0$ if and only if $\aval''(\aclock)=0$ for all $\aval'' \in \areg$,
and
$\aval'(\aclock)=0$ if and only if $\aval''(\aclock)=0$ for all $\aval'' \in \areg'$.
Hence $\aclockset^0_\aval = \aclockset^0_\areg$ and $\aclockset^0_{\aval'} = \aclockset^0_{\areg'}$.
Given that either $\aval'(\aclock)=\aval(\aclock)$ or $\aval'(\aclock)=0$ for each $\aclock \in \clocks$,
we have that $\aclockset \in \reset{\aval}{\aval'}$ if and only if 
$\aclockset^0_{\aval'} \setminus \aclockset^0_\aval \subseteq \aclockset \subseteq \aclockset^0_{\aval'}$.
Similarly, $\aclockset \in \reset{\areg}{\areg'}$ if and only if 
$\aclockset^0_{\areg'} \setminus \aclockset^0_\areg \subseteq \aclockset \subseteq \aclockset^0_{\areg'}$.
Therefore we have that $\aclockset \in \reset{\aval}{\aval'}$ if and only if 
$\aclockset^0_{\aval'} \setminus \aclockset^0_\aval \subseteq \aclockset \subseteq \aclockset^0_{\aval'}$
if and only if 
$\aclockset^0_{\areg'} \setminus \aclockset^0_\areg \subseteq \aclockset \subseteq \aclockset^0_{\areg'}$
if and only if 
$\aclockset \in \reset{\areg}{\areg'}$.
Hence $\reset{\aval}{\aval'} = \reset{\areg}{\areg'}$.
%
%
%
%
%
%
\qed
\end{proof}

A \emph{set of weights} is a finite set $\set{\weight_i}{i \in I}$ 
such that $\weight_i \in (0,1]$ for each $i \in I$ and $\sum_{i \in I} \weight_i = 1$.
In the following, we use an interpretation of valuations and corner points 
as points in $\nnr^{|\clocks|}$-space,
allowing the use of operations such as $\weight \cdot \aval$ and $\aval + \aval'$
(interpreted as $(\weight \cdot \aval)(\aclock) = \weight \cdot \aval(\aclock)$
and $(\aval + \aval')(\aclock) = \aval(\aclock) + \aval'(\aclock)$ for all clocks $\aclock \in \clocks$, 
respectively).

\begin{lemma}\label{lem:existence_weights}
Let $\aval \in \valuations$, let $k \in \Nset$ 
and let $\areg \in \regs{\maxbound}{k}$ be the unique $k$-region such that $\aval \in \areg$.
Then there exists a set of weights $\set{\weight_\acp}{\acp \in \cpsofreg{\areg}}$
such that 
$\aval = \sum_{\acp \in \cpsofreg{\areg}} \weight_\acp \cdot \acp$.
\end{lemma}
\begin{proof}
Observe that the convex hull of corner points $\cpsofreg{\areg}$
corresponds to a superset of the valuations contained in $\areg$.
Hence, given that $\aval \in \areg$,
we have that $\aval$ is in the set of valuations induced by the convex hull of $\cpsofreg{\areg}$,
and hence there exists $\set{\weight_\acp}{\acp \in \cpsofreg{\areg}}$ 
with the required property.
%
%
%
%
\qed
\end{proof}

%
In the following, for a state $(\aloc,\aval) \in \states$ of $\semPTS{\acdpta}$, 
we use $\inrg{\aloc}{\val}{\maxbound}{k}$ to denote the unique pair $(\aloc',\areg) \in \locs \times \regs{\maxbound}{k}$
such that $\aloc = \aloc'$ and $\aval \in \areg$.

\begin{lemma}\label{lem:weight_pedge}
Let $k \in \Nset$,
let $\areg \in \regs{\maxbound}{k}$ be the $k$-region such that $\aval \in \areg$,
and let $(\aloc,\g,\pd) \in \pedges$ be a probabilistic edge such that $\aval \vinz \g$.
Then there exists a set of weights $\set{\weight_\acp}{\acp \in \cpsofreg{\areg}}$ 
such that, for any $(\aclockset,\aloc') \in 2^\clocks \times \locs$:
\[
\adtpara{\pd}{\aval}(\aclockset,\aloc') 
= 
\sum_{\acp \in \cpsofreg{\areg}} \weight_\acp \cdot \adtpara{\pd}{\acp}(\aclockset,\aloc') \; .
\]
\end{lemma}
\begin{proof}
Let $\set{\weight_\acp}{\acp \in \cpsofreg{\areg}}$ be the set of weights such that 
$\aval = \sum_{\acp \in \cpsofreg{\areg}} \weight_\acp \cdot \acp$,
which exists by Lemma~\ref{lem:existence_weights}.
Let $\outcome = (\aclockset,\aloc') \in 2^\clocks \times \locs$.
For clock $\aclock \in \clocks$,
we use $\aninterval_{\aval}$ to denote 
the interval of the partition $\apartition{\apedge,\outcome}{\aclock}$
such that $\aval(\aclock) \in \aninterval_{\aval}$,
and use $\fc{\apedge,\outcome}{\aclock}{\aninterval_{\aval}}$
and $\fd{\apedge,\outcome}{\aclock}{\aninterval_{\aval}}$ to denote the 
constants such that $\plf{\apedge,\outcome}{\aclock}(\gamma) = 
\fc{\apedge,\outcome}{\aclock}{\aninterval_{\aval}} 
+ \fd{\apedge,\outcome}{\aclock}{\aninterval_{\aval}} \cdot \gamma$
if $\gamma \in \aninterval_{\aval}$.
Then we have:
\[
\begin{array}{rclr}
\adtpara{\pd}{\aval}(\outcome) 
& = & 
\displaystyle{
\sum_{\aclock \in \clocks} \plf{\apedge,\outcome}{\aclock}(\aval(\aclock))
}
&
\\
& = &
\displaystyle{
\sum_{\aclock \in \clocks}
(\fc{\apedge,\outcome}{\aclock}{\aninterval_{\aval}} 
+ \fd{\apedge,\outcome}{\aclock}{\aninterval_{\aval}} \cdot \aval(\aclock))
}
&
\\
& = &
\displaystyle{
\sum_{\aclock \in \clocks}
(\fc{\apedge,\outcome}{\aclock}{\aninterval_{\aval}} 
+ \fd{\apedge,\outcome}{\aclock}{\aninterval_{\aval}} 
\cdot \sum_{\acp \in \cpsofreg{\areg}} \weight_\acp \cdot \acp(\aclock))
}
&
\\
& = &
\displaystyle{
\sum_{\aclock \in \clocks}
\fc{\apedge,\outcome}{\aclock}{\aninterval_{\aval}} 
+ 
\sum_{\aclock \in \clocks}
\fd{\apedge,\outcome}{\aclock}{\aninterval_{\aval}} 
\cdot \sum_{\acp \in \cpsofreg{\areg}} \weight_\acp \cdot \acp(\aclock)
}
&
\\
& = &
\displaystyle{
\sum_{\acp \in \cpsofreg{\areg}}
\weight_\acp 
\sum_{\aclock \in \clocks}
\fc{\apedge,\outcome}{\aclock}{\aninterval_{\aval}} 
+
\sum_{\acp \in \cpsofreg{\areg}} 
\weight_\acp 
\sum_{\aclock \in \clocks}
\fd{\apedge,\outcome}{\aclock}{\aninterval_{\aval}} 
\cdot  \acp(\aclock)
}
&
\mbox{(from $\displaystyle{\!\! \sum_{\acp \in \cpsofreg{\areg}} \!\! \weight_\acp = 1}$)}
\\
& = &
\displaystyle{
\sum_{\acp \in \cpsofreg{\areg}}
\weight_\acp 
(
\sum_{\aclock \in \clocks}
\fc{\apedge,\outcome}{\aclock}{\aninterval_{\aval}} 
+ 
\sum_{\aclock \in \clocks}
\fd{\apedge,\outcome}{\aclock}{\aninterval_{\aval}} 
\cdot \acp(\aclock)
)
\; .
}
&
\end{array}
\]
Recall that $\aninterval_{\aval}$ has natural-numbered endpoints,
and that $\acp(\aclock)$ is a rational number.
Note that it may be the case that $\aninterval_{\aval}$
is open or half-open,
and hence may not include $\acp(\aclock)$.
Given that $\plf{\apedge,\outcome}{\aclock}$ is a continuous function,
we have that $\plf{\apedge,\outcome}{\aclock}(\gamma) = 
\fc{\apedge,\outcome}{\aclock}{\aninterval_{\aval}} 
+ \fd{\apedge,\outcome}{\aclock}{\aninterval_{\aval}} \cdot \gamma$
for all $\gamma$ in the closure of $\aninterval_{\aval}$.
Given that $\acp(\aclock)$ must belong to the closure of $\aninterval_{\aval}$,
we conclude the following:
\begin{eqnarray*}
\sum_{\acp \in \cpsofreg{\areg}}
\weight_\acp 
(
\sum_{\aclock \in \clocks}
\fc{\apedge,\outcome}{\aclock}{\aninterval_{\aval}} 
+ 
\sum_{\aclock \in \clocks}
\fd{\apedge,\outcome}{\aclock}{\aninterval_{\aval}} 
\cdot \acp(\aclock)
)
& = &
\sum_{\acp \in \cpsofreg{\areg}}
\weight_\acp 
\sum_{\aclock \in \clocks} \plf{\apedge,\outcome}{\aclock}(\acp(\aclock))
\\
& = &
\sum_{\acp \in \cpsofreg{\areg}}
\weight_\acp \cdot \adtpara{\pd}{\acp}(\outcome) \; .
\end{eqnarray*}
Hence we have shown that 
$\adtpara{\pd}{\aval}(\outcome) 
= 
\sum_{\acp \in \cpsofreg{\areg}} \weight_\acp \cdot \adtpara{\pd}{\acp}(\outcome)$,
which concludes the proof.
\qed
\end{proof}

\begin{lemma}\label{lem:one-step_pedge}
Let $(\aloc,\aval) \in \states$ be a state,
let $k \in \Nset$,
and let $\areg \in \regs{\maxbound}{k}$ be the $k$-region such that $\aval \in \areg$.
For each transition $((\aloc,\aval),(\aloc,\g,\pd),\adist) \in \ptstrans$ of $\semPTS{\acdpta}$,
there exists a set of transitions 
$\set{(\inrg{\aloc}{\val}{\maxbound}{k},(\acp,(\aloc,\g,\pd)),\anrgdist_\acp)}{\acp \in \cpsofreg{\areg}} 
\subseteq 
\rgtrans{\acdpta}{k}$ of $\acdrg{\acdpta}{k}$
and weights $\set{\weight_\acp}{\acp \in \cpsofreg{\areg}}$ 
such that, for each state $(\aloc',\aval') \in \states$:
\[
\adist(\aloc',\aval') 
= 
\sum_{\acp \in \cpsofreg{\areg}} \weight_\acp \cdot \anrgdist_\acp(\inrg{\aloc'}{\val'}{\maxbound}{k}) \; .
\]
\end{lemma}
\begin{proof}
Let $\set{\weight_\acp}{\acp \in \cpsofreg{\areg}}$ be the set of weights such that 
$\aval = \sum_{\acp \in \cpsofreg{\areg}} \weight_\acp \cdot \acp$,
which exists by Lemma~\ref{lem:existence_weights},
and let $\areg,\areg' \in \regs{\maxbound}{k}$ be the $k$-regions 
such that $\aval \in \areg$ and $\aval \in \areg'$.
By definition of $\semPTS{\acdpta}$, we have:
\[
\begin{array}{rclr}
\adist(\aloc',\aval') 
& = &
\displaystyle{
\sum_{\aclockset \in \reset{\aval}{\aval'}} \adtpara{\pd}{\aval}(\aclockset,\aloc')
}
&
\\
& = &
\displaystyle{
\sum_{\aclockset \in \reset{\aval}{\aval'}} 
\sum_{\acp \in \cpsofreg{\areg}} \weight_\acp \cdot \adtpara{\pd}{\acp}(\aclockset,\aloc')
}
&
\mbox{(by Lemma~\ref{lem:weight_pedge})}
\\
& = &
\displaystyle{
\sum_{\aclockset \in \reset{\areg}{\areg'}} 
\sum_{\acp \in \cpsofreg{\areg}} \weight_\acp \cdot \adtpara{\pd}{\acp}(\aclockset,\aloc')
}
&
\mbox{(by Lemma~\ref{lem:resets})}
\\
& = &
\displaystyle{
\sum_{\acp \in \cpsofreg{\areg}} 
\weight_\acp 
\sum_{\aclockset \in \reset{\areg}{\areg'}} 
\adtpara{\pd}{\acp}(\aclockset,\aloc')
}
&
\\
& = &
\displaystyle{
\sum_{\acp \in \cpsofreg{\areg}} 
\weight_\acp \cdot 
\anrgdist_i(\inrg{\aloc'}{\aval'}{\maxbound}{k})
\; .
}
&
\end{array}
\]
\begin{flushright}
\qed
\end{flushright}
\end{proof}

The next lemma follows from standard non-probabilistic reasoning on the region graph.

\begin{lemma}\label{lem:one-step_time}
\sloppypar{
Let $(\aloc,\aval) \in \states$ be a state,
and let $k \in \Nset$.
For each transition $((\aloc,\aval),\adelay,\dirac{(\aloc,\aval+\adelay)}) \in \ptstrans$ of $\semPTS{\acdpta}$,
there exists a transition 
$(\inrg{\aloc}{\val}{\maxbound}{k},\rsuccact,\dirac{\inrg{\aloc}{\aval+\adelay}{\maxbound}{k}}) 
\in \rgtrans{\acdpta}{k}$ 
of $\acdrg{\acdpta}{k}$.
}
\end{lemma}

The following lemma specifies that, for any transition of $\semPTS{\acdpta}$,
any two distinct states within its distribution's support set
belong to different $k$-regions.

\begin{lemma}\label{lem:concrete_support}
Let $(\aloc,\aval) \in \states$ be a state,
let $k \in \Nset$,
and let $((\aloc,\aval),(\aloc,\g,\pd),\adist) \in \ptstrans$ be a transition of $\semPTS{\acdpta}$.
For each pair $(\aloc_1,\aval_1), (\aloc_2,\aval_2) \in \support(\adist)$ 
such that $(\aloc_1,\aval_1) \neq (\aloc_2,\aval_2)$,
we have $\inrg{\aloc_1}{\aval_1}{\maxbound}{k} \neq \inrg{\aloc_2}{\aval_2}{\maxbound}{k}$.
\end{lemma}
\begin{proof}
Let $(\aloc_1,\aval_1), (\aloc_2,\aval_2) \in \support(\adist)$ such that $(\aloc_1,\aval_1) \neq (\aloc_2,\aval_2)$.
First observe that if $\aloc_1 \neq \aloc_2$ then trivially
$\inrg{\aloc_1}{\aval_1}{\maxbound}{k} \neq \inrg{\aloc_2}{\aval_2}{\maxbound}{k}$.
Now consider the case in which $\aloc_1 = \aloc_2$ and $\aval_1 \neq \aval_2$.
we must have $\aval_1 \neq \aval_2$.
Note that $\aval_1=\aval[\aclockset_1:=0]$ and $\aval_2=\aval[\aclockset_2:=0]$
for clock sets $\aclockset_1,\aclockset_2 \subseteq \clocks$.
Hence $\aval_1$ and $\aval_2$ differ only in terms of which clocks are equal to 0.
Intuitively, by the definition of $k$-regions, 
any two valuations that differ only in terms of which clocks are equal to 0 belong to different $k$-regions.
For completeness, we now explain this formally.
Denote the sets of clocks that are equal to 0 in $\aval_1$ by $\aclockset_1'$
and in $\aval_2$ by $\aclockset_2'$ 
(note that $\aclockset_1 \subseteq \aclockset_1'$, $\aclockset_2 \subseteq \aclockset_2'$
and that $\aclockset_1' \neq \aclockset_2'$ because $\aval_1 \neq \aval_2$).
Let the $k$-region component of $\inrg{\aloc_1}{\val_1}{\maxbound}{k}$ 
be denoted by $(\anatval_1,[\aclockset_{1,0}, \aclockset_{1,1},..., \aclockset_{1,n_1}])$
and let the $k$-region component of $\inrg{\aloc_2}{\val_2}{\maxbound}{k}$ 
be denoted by $(\anatval_2,[\aclockset_{2,0}, \aclockset_{2,1},..., \aclockset_{2,n_2}])$.
Given that $\aclockset_1' \neq \aclockset_2'$,
either there exists clock $\aclock \in \aclockset_1' \setminus \aclockset_2'$
such that $\anatval_1(\aclock)=0$ and $\aclock \in \aclockset_{1,0}$
but either $\anatval_2(\aclock) \neq 0$ or $\aclock \not\in \aclockset_{2,0}$,
or there exists clock $\aclock \in \aclockset_2' \setminus \aclockset_1'$
such that $\anatval_2(\aclock)=0$ and $\aclock \in \aclockset_{2,0}$
but either $\anatval_1(\aclock) \neq 0$ or $\aclock \not\in \aclockset_{1,0}$.
Hence we have either $\anatval_1 \neq \anatval_2$ or $\aclockset_{1,0} \neq \aclockset_{2,0}$,
and therefore $\inrg{\aloc_1}{\aval_1}{\maxbound}{k} \neq \inrg{\aloc_2}{\aval_2}{\maxbound}{k}$.
\qed
\end{proof}

Lemma~\ref{lem:concrete_support} specifies that,
for each transition $((\aloc,\aval),\anaction,\adist) \in \ptstrans$ of $\semPTS{\acdpta}$
and for each $(\aloc',\areg) \in \rgstates{\acdpta}{k}$,
there exists at most one valuation $\aval' \in \areg$
such that $(\aloc',\aval') \in \support(\adist)$.
If such a valuation $\aval'$ exists, we set $\aval_{\adist,(\aloc',\areg)} = \aval'$,
otherwise $\aval_{\adist,(\aloc',\areg)}$ can be set to an arbitrary valuation.
%
From this fact, together with Lemma~\ref{lem:one-step_pedge} and Lemma~\ref{lem:one-step_time},
we obtain the following lemma.

\begin{lemma}\label{lem:one-step}
Let $(\aloc,\aval) \in \states$ be a state,
and let $k \in \Nset$.
For each transition $((\aloc,\aval),\anaction,\adist) \in \ptstrans$ of $\semPTS{\acdpta}$,
there exists a combined transition 
$(\set{(\inrg{\aloc}{\aval}{\maxbound}{k},\anaction_i,\anrgdist_i)}{i \in I}, \set{\ctweight_i}{i \in I})$ 
of $\acdrg{\acdpta}{k}$
such that, for each $(\aloc',\areg') \in \rgstates{\acdpta}{k}$,
we have:
\begin{enumerate}
\item
$\adist(\aloc',\aval_{\adist,(\aloc',\areg')}) = \sum_{i \in I} \ctweight_i \cdot \anrgdist_i(\aloc',\areg')$.
\item
$\sum_{\aval' \in \areg'} \adist(\aloc',\aval') = \sum_{i \in I} \ctweight_i \cdot \anrgdist_i(\aloc',\areg')$.
\end{enumerate}
\end{lemma}
\begin{proof}
We first consider part (1).
Let $\areg \in \regs{\maxbound}{k}$ be the unique region such that $\aval \in \areg$.
We consider the following two cases. 

{\sl Case $\anaction \in \pedges$.}
Let $\apedge = \anaction$.
By Lemma~\ref{lem:one-step_pedge}, there exist 
$\set{((\aloc,\areg),(\acp,\apedge),\anrgdist_\acp)}{\acp \in \cpsofreg{\areg}} 
\subseteq 
\rgtrans{\acdpta}{k}$ of $\acdrg{\acdpta}{k}$
and weights $\set{\weight_\acp}{\acp \in \cpsofreg{\areg}}$ 
such that
$\adist(\aloc',\aval_{\adist,(\aloc',\areg')})  = 
\sum_{\acp \in \cpsofreg{\areg}} \weight_\acp \cdot \anrgdist_\acp(\inrg{\aloc'}{\aval_{\adist,(\aloc',\areg')}}{\maxbound}{k})$.
Hence we let $I = \cpsofreg{\areg}$ and 
$\ctweight_\acp = \weight_\acp$ for each $\acp \in \cpsofreg{\areg}$,
concluding that 
$\adist(\aloc',\aval_{\adist,(\aloc',\areg')}) = \sum_{i \in I} \ctweight_i \cdot \anrgdist_i(\aloc',\areg')$.

{\sl Case $\anaction \in \nnr$.}
Let $\adelay = \anaction$.
Note that, by definition of $\semPTS{\acdpta}$,
for the unique $(\aloc',\areg') \in \rgstates{\acdpta}{k}$ 
such that $\aloc = \aloc'$ and $\aval+\adelay \in \areg'$,
we must have $\aval_{\adist,(\aloc',\areg')} = \aval+\adelay$, 
i.e., $\adist(\aloc',\aval_{\adist,(\aloc',\areg')}) = \adist(\aloc',\aval+\adelay) = 1$.
By Lemma~\ref{lem:one-step_time},
there exists 
$((\aloc,\areg),\rsuccact,\dirac{\inrg{\aloc}{\aval+\adelay}{\maxbound}{k}}) 
\in \rgtrans{\acdpta}{k}$:
hence we let $|I|=1$ and let $\set{\ctweight_i}{i \in I}$ be the set containing a single weight equal to 1.
Then we conclude that 
$\adist(\aloc',\aval_{\adist,(\aloc',\areg')}) 
= 
\adist(\aloc',\aval+\adelay)
=
1
=
\dirac{\inrg{\aloc}{\aval+\adelay}{\maxbound}{k}}(\inrg{\aloc}{\aval+\adelay}{\maxbound}{k})
= 
\sum_{i \in I} \ctweight_i \cdot \anrgdist_i(\aloc',\areg')$.

Part (2) of the lemma then follows from the fact that,
for $(\aloc',\areg') \in \rgstates{\acdpta}{k}$ such that 
there exists a valuation $\aval' \in \areg'$ with $(\aloc',\aval') \in \support(\adist)$,
we have $\sum_{\aval'' \in \areg'} \adist(\aloc',\aval'') = \adist(\aloc',\aval_{\adist,(\aloc',\areg')})$.
\qed
\end{proof}

Consider equivalence $\anequiv \subseteq (\states \uplus \rgstates{\acdpta}{k})^2$ 
over the states of
the disjoint union of $\semPTS{\acdpta}$ and $\acdrg{\acdpta}{k}$
defined as the smallest equivalence satisfying the following conditions:
\begin{itemize}
\item
for states $(\aloc,\aval),(\aloc',\aval') \in \states$,
we have $(\aloc,\aval) \anequiv (\aloc',\aval')$ 
if $\inrg{\aloc}{\val}{\maxbound}{k} = \inrg{\aloc'}{\val'}{\maxbound}{k}$
(i.e., $\aloc=\aloc'$, and $\aval$ and $\aval'$ belong to the same $k$-region in $\regs{\maxbound}{k}$);
\item
for $(\aloc,\aval) \in \states$, $(\aloc',\areg) \in \rgstates{\acdpta}{k}$,
we have $(\aloc,\aval) \anequiv (\aloc',\areg)$ 
if $\inrg{\aloc}{\val}{\maxbound}{k} = (\aloc',\areg)$
(i.e., $\aloc=\aloc'$ and $\aval$ belongs to $\areg$).
\end{itemize}
Then the following corollary is a direct consequence of part (2) of Lemma~\ref{lem:one-step}.

\begin{corollary}\label{cor:sum_to_distsum}
\sloppypar{
Let $(\aloc,\aval) \in \states$ be a state,
and let $k \in \Nset$.
For each transition $((\aloc,\aval),\anaction,\adist) \in \ptstrans$ of $\semPTS{\acdpta}$,
there exists a combined transition 
$(\set{(\inrg{\aloc}{\val}{\maxbound}{k},\anaction_i,\anrgdist_i)}{i \in I}, \set{\ctweight_i}{i \in I})$ 
of $\acdrg{\acdpta}{k}$
such that
$\adist \equiv \distsum_{i \in I} \ctweight_i \cdot \anrgdist_i$
and either
$\anaction_i = \rsuccact$ for all $i \in I$ if $\anaction \in \nnr$,
and $\set{\anaction_i}{i \in I} \subseteq \cps{\maxbound}{k} \times \pedges$ otherwise.
}
\end{corollary}


We now proceed to the proof of part~(1) of Proposition~\ref{prop:approx}.

\begin{proof}[of part~(1) of Proposition~\ref{prop:approx}]
Consider the relation $\proposedpsim \subseteq (\states \uplus \rgstates{\acdpta}{k})^2$
such that $\proposedpsim$ is the smallest relation satisfying the following property:
for $(\aloc,\aval) \in \states$, $(\aloc',\areg) \in \rgstates{\acdpta}{k}$,
we have $(\aloc,\aval) \proposedpsim (\aloc',\areg)$ 
if $\inrg{\aloc}{\val}{\maxbound}{k} = (\aloc',\areg)$.
By Corollary~\ref{cor:sum_to_distsum},
$\proposedpsim$ is a probabilistic simulation respecting $\anequiv$ and $\{ \rsuccact \} \cup \nnr$.
Then, by Proposition~\ref{prop:hahn},
we have that 
$\maxval{\semPTS{\acdpta}}{\finalstates}{\altstrategies{\nnr}} 
\leq 
\maxval{\acdrg{\acdpta}{k}}{\finalregs{\maxbound}{k}}{\altstrategies{\{ \rsuccact \}}}$
and $\minval{\semPTS{\acdpta}}{\finalstates}{\altstrategies{\nnr}} 
\geq 
\minval{\acdrg{\acdpta}{k}}{\finalregs{\maxbound}{k}}{\altstrategies{\{ \rsuccact \}}}$.
Noting that $\cdptastrategies = \altstrategies{\nnr}$
and $\rgstrategies{\acdpta}{k} = \altstrategies{\{ \rsuccact \}}$,
we have that 
$\maxval{\semPTS{\acdpta}}{\finalstates}{\cdptastrategies}
\leq
\maxval{\acdrg{\acdpta}{k}}{\finalregs{\maxbound}{k}}{\rgstrategies{\acdpta}{k}}$
and
$\minval{\semPTS{\acdpta}}{\finalstates}{\cdptastrategies} \geq 
\minval{\acdrg{\acdpta}{k}}{\finalregs{\maxbound}{k}}{\rgstrategies{\acdpta}{k}}$.
\qed
\end{proof}


\subsection{Approximating granularity $2k$ by granularity $k$}

For $2k$-region $\areg \in \regs{\maxbound}{2k}$ and $k$-region $\areg' \in  \regs{\maxbound}{k}$,
we write $\areg \subseteq \areg'$
if every valuation that is contained in $\areg$ is also contained in $\areg'$
(i.e., 
if $\{\aval \in \valuations \setsep \aval \in \areg\} \subseteq \{\aval \in \valuations \setsep \aval \in \areg'\}$).
Note that, for a given $2k$-region $\areg \in \regs{\maxbound}{2k}$ 
there is exactly one $k$-region $\areg' \in \regs{\maxbound}{k}$ 
such that $\areg \subseteq \areg'$.
In the following, given the $2k$-region $\areg$,
we use $\regcont{\areg}{k}$ to denote the unique $k$-region such that $\areg \subseteq \regcont{\areg}{k}$.
We now adapt Lemma~\ref{lem:resets} to the case of $2k$-regions and $k$-regions:
that is, 
the sets of clocks that, when reset to 0, 
are used to transform $2k$-region $\areg$ to $2k$-region  $\areg'$
are the same as the sets of clocks used to transform  
the $k$-region containing the $2k$-region $\areg$
to the $k$-region containing the $2k$-region $\areg'$.
The proof of the lemma proceeds in an analogous manner to that of Lemma~\ref{lem:resets},
and is therefore omitted.

\begin{lemma}\label{lem:resets_k_2k}
Let $k \in \Nset$ and 
let $\areg_{2k},\areg_{2k}' \in \regs{\maxbound}{2k}$
such that $\areg_{2k}' = \areg_{2k}[\aclockset:=0]$ for some $\aclockset \subseteq \clocks$.
Using $\areg_k, \areg_k' \in \regs{\maxbound}{k}$ to denote the unique $k$-regions 
such that $\areg_{2k} \subseteq \areg_k$ and $\areg_{2k}' \subseteq \areg_k'$,
we have $\reset{\areg_{2k}}{\areg_{2k}'} = \reset{\areg_k}{\areg_k'}$.
\end{lemma}

The following result specifies that
every corner point of $\areg  \in \regs{\maxbound}{2k}$
is either a corner point of $\regcont{\areg}{k}$ 
or can be obtained from a weighted combination of corner points of $\regcont{\areg}{k}$.

\begin{lemma}\label{lem:corner_k_2k}
Let $k \in \Nset$ and
let $\areg \in \regs{\maxbound}{2k}$.
For each corner point $\acp \in \cpsofreg{\areg}$,
there exist a set of weights $\set{\weight_{\acp'}}{\acp' \in \cpsofreg{\regcont{\areg}{k}}}$
such that $\acp = \sum_{\acp' \in \cpsofreg{\regcont{\areg}{k}}} \weight_{\acp'} \cdot \acp'$.
\end{lemma}
\begin{proof}
Note that the convex hull of corner points in $\cpsofreg{\regcont{\areg}{k}}$
is a superset of the convex hull of corner points in $\cpsofreg{\areg}$.
Hence, any corner point $\acp \in \cpsofreg{\areg}$
is in the set of valuations induced by the convex hull of $\cpsofreg{\regcont{\areg}{k}}$,
and hence there exists the required $\set{\weight_{\acp'}}{\acp' \in \cpsofreg{\regcont{\areg}{k}}}$
such that $\acp = \sum_{\acp' \in \cpsofreg{\regcont{\areg}{k}}} \weight_{\acp'} \cdot \acp'$.
\qed
\end{proof}
We note that the corner points of $\areg \in \regs{\maxbound}{2k}$
are either also corner points of the unique $\areg'  \in \regs{\maxbound}{k}$ such that $\areg \subseteq \areg'$,
or they are mid-points of edges of the polyhedron induced by the convex hull of the corner points of $\areg'$.

Lemma~\ref{lem:corner_k_2k} allows us to state the following lemma 
(which is an analogue of Lemma~\ref{lem:weight_pedge}).

\begin{lemma}\label{lem:weight_pedge_k_2k}
Let $k \in \Nset$,
let $\areg \in \regs{\maxbound}{2k}$,
let $(\aloc,\g,\pd) \in \pedges$ be a probabilistic edge such that $\areg \vinz \g$,
and let $\acp \in \cpsofreg{\areg}$ be a corner point of $\areg$.
Then there exists a set of weights $\set{\weight_{\acp'}}{\acp' \in \cpsofreg{\regcont{\areg}{k}}}$ 
such that, for any $(\aclockset,\aloc') \in 2^\clocks \times \locs$, 
we have:
\[
\adtpara{\pd}{\acp}(\aclockset,\aloc') 
= 
\sum_{\acp' \in \cpsofreg{\regcont{\areg}{k}}} \weight_{\acp'} \cdot \adtpara{\pd}{\acp'}(\aclockset,\aloc') \; .
\]
\end{lemma}
\begin{proof}
By Lemma~\ref{lem:corner_k_2k},
it is possible that $\acp \in \cpsofreg{\regcont{\areg}{k}}$,
in which case we let $\weight_{\acp} = 1$ and trivially we have:
\[
\adtpara{\pd}{\acp}(\aclockset,\aloc')  
= \weight_{\acp} \cdot \adtpara{\pd}{\acp}(\aclockset,\aloc')
= \sum_{\acp' \in \cpsofreg{\regcont{\areg}{k}}} \weight_{\acp'} \cdot \adtpara{\pd}{\acp'}(\aclockset,\aloc') \; .
\]

Now consider the case in which $\acp \not\in \cpsofreg{\regcont{\areg}{k}}$.
We proceed in a similar manner to the proof of Lemma~\ref{lem:weight_pedge}.
By Lemma~\ref{lem:corner_k_2k},
we have the existence of a set of weights $\set{\weight_{\acp'}}{\acp' \in \cpsofreg{\regcont{\areg}{k}}}$ such that 
such that $\acp = \sum_{\acp' \in \cpsofreg{\regcont{\areg}{k}}} \weight_{\acp'} \cdot \acp'$.
Let $\outcome = (\aclockset,\aloc') \in 2^\clocks \times \locs$.
For clock $\aclock \in \clocks$,
we define $\aninterval_\acp$ as 
the interval of the partition $\apartition{\apedge,\outcome}{\aclock}$
such that $\acp(\aclock) \in \aninterval_\acp$,
and use $\fc{\apedge,\outcome}{\aclock}{\aninterval_\acp}$
and $\fd{\apedge,\outcome}{\aclock}{\aninterval_\acp}$ to denote the 
constants such that $\plf{\apedge,\outcome}{\aclock}(\gamma) = 
\fc{\apedge,\outcome}{\aclock}{\aninterval_\acp} 
+ \fd{\apedge,\outcome}{\aclock}{\aninterval_\acp} \cdot \gamma$
if $\gamma \in \aninterval_\acp$.
Then we have:
\begin{eqnarray*}
\adtpara{\pd}{\acp}(\outcome) 
& = & 
\displaystyle{
\sum_{\aclock \in \clocks} \plf{\apedge,\outcome}{\aclock}(\acp(\aclock))
}
\\
& = &
\displaystyle{
\sum_{\aclock \in \clocks}
(\fc{\apedge,\outcome}{\aclock}{\aninterval_\acp} 
+ \fd{\apedge,\outcome}{\aclock}{\aninterval_\acp} \cdot \acp(\aclock))
}
\\
& = &
\displaystyle{
\sum_{\aclock \in \clocks}
(\fc{\apedge,\outcome}{\aclock}{\aninterval_\acp} 
+ \fd{\apedge,\outcome}{\aclock}{\aninterval_\acp} 
\cdot \sum_{\acp' \in \cpsofreg{\regcont{\areg}{k}}} \weight_{\acp'} \cdot \acp'(\aclock))
}
\\
& = &
\displaystyle{
\sum_{\aclock \in \clocks}
\fc{\apedge,\outcome}{\aclock}{\aninterval_\acp} 
+ 
\sum_{\aclock \in \clocks}
\fd{\apedge,\outcome}{\aclock}{\aninterval_\acp} 
\cdot \sum_{\acp' \in \cpsofreg{\regcont{\areg}{k}}} \weight_{\acp'} \cdot \acp'(\aclock)
}
\\
& = &
\displaystyle{
\sum_{\acp' \in \cpsofreg{\regcont{\areg}{k}}}
\weight_{\acp'} 
\sum_{\aclock \in \clocks}
\fc{\apedge,\outcome}{\aclock}{\aninterval_\acp} 
+
\sum_{\acp' \in \cpsofreg{\regcont{\areg}{k}}} 
\weight_{\acp'} 
\sum_{\aclock \in \clocks}
\fd{\apedge,\outcome}{\aclock}{\aninterval_\acp} 
\cdot  \acp'(\aclock)
}
\\
& = &
\displaystyle{
\sum_{\acp' \in \cpsofreg{\regcont{\areg}{k}}}
\weight_{\acp'} 
(
\sum_{\aclock \in \clocks}
\fc{\apedge,\outcome}{\aclock}{\aninterval_\acp} 
+ 
\sum_{\aclock \in \clocks}
\fd{\apedge,\outcome}{\aclock}{\aninterval_\acp} 
\cdot \acp'(\aclock)
)
}
\\
& = &
\displaystyle{
\sum_{\acp' \in \cpsofreg{\regcont{\areg}{k}}}
\weight_{\acp'} 
\sum_{\aclock \in \clocks} \plf{\apedge,\outcome}{\aclock}(\acp'(\aclock))
}
\\
& = &
\displaystyle{
\sum_{\acp' \in \cpsofreg{\regcont{\areg}{k}}}
\weight_{\acp'} \cdot \adtpara{\pd}{\acp'}(\outcome) \; ,
}
\end{eqnarray*}
(where the fifth equation follows from $\sum_{\acp' \in \cpsofreg{\regcont{\areg}{k}}} \weight_{\acp'}  = 1$,
and the penultimate equation follows from the fact that $\plf{\apedge,\outcome}{\aclock}$
is a continuous function, as in the proof of Lemma~\ref{lem:weight_pedge})
which concludes the proof.
\qed
\end{proof}

\begin{lemma}\label{lem:one-step_pedge_k_2k}
Let $k \in \Nset$
and $\areg \in \regs{\maxbound}{2k}$.
For each transition $((\aloc,\areg),(\acp,(\aloc,\g,\pd)),\anrgdist) \in \rgtrans{\acdpta}{2k}$ 
of $\acdrg{\acdpta}{2k}$,
there exists a set of transitions 
$\set{(\aloc,\regcont{\areg}{k}),(\acp',(\aloc,\g,\pd)),\anrgdist_{\acp'})}{\acp' \in \cpsofreg{\regcont{\areg}{k}}} 
\subseteq 
\rgtrans{\acdpta}{k}$ of $\acdrg{\acdpta}{k}$
and weights $\set{\weight_{\acp'}}{\acp' \in \cpsofreg{\regcont{\areg}{k}}}$ 
such that, for each state $(\aloc',\areg') \in \rgstates{\acdpta}{2k}$,
we have:
\[
\anrgdist(\aloc',\areg') 
= 
\sum_{\acp' \in \cpsofreg{\regcont{\areg}{k}}} 
\weight_{\acp'} \cdot \anrgdist_{\acp'}(\aloc',\regcont{\areg'}{k}) 
\; .
\]
\end{lemma}
\begin{proof}
We proceed in a similar manner to the proof of Lemma~\ref{lem:one-step_pedge}.
Let $\set{\weight_{\acp'}}{\acp' \in \cpsofreg{\regcont{\areg}{k}}}$ be the set of weights such that 
$\acp = \sum_{\acp' \in \cpsofreg{\regcont{\areg}{k}}} \weight_{\acp'} \cdot \acp'$,
which exists by Lemma~\ref{lem:corner_k_2k}.
Then for each $(\aloc',\areg') \in \rgstates{\acdpta}{2k}$,
by the definition of $\acdrg{\acdpta}{2k}$, we have:
\[
\begin{array}{rcll}
\anrgdist(\aloc',\areg') 
& = & 
\displaystyle{
\sum_{\aclockset \in \reset{\areg}{\areg'}} \adtpara{\pd}{\acp}(\aclockset,\aloc')
}
&
\\
& = &
\displaystyle{
\sum_{\aclockset \in \reset{\areg}{\areg'}} 
\sum_{\acp' \in \cpsofreg{\regcont{\areg}{k}}} 
\weight_{\acp'} \cdot \adtpara{\pd}{\acp'}(\aclockset,\aloc') 
}
&
\mbox{ (by Lemma~\ref{lem:weight_pedge_k_2k})}
\\
& = &
\displaystyle{
\sum_{\aclockset \in \reset{\regcont{\areg}{k}}{\regcont{\areg'}{k}}} 
\sum_{\acp' \in \cpsofreg{\regcont{\areg}{k}}} 
\weight_{\acp'} \cdot \adtpara{\pd}{\acp'}(\aclockset,\aloc') 
}
&
\mbox{ (by Lemma~\ref{lem:resets_k_2k})}
\\
& = &
\displaystyle{
\sum_{\acp' \in \cpsofreg{\regcont{\areg}{k}}} 
\weight_{\acp'} \cdot 
\sum_{\aclockset \in \reset{\regcont{\areg}{k}}{\regcont{\areg'}{k}}} 
\adtpara{\pd}{\acp'}(\aclockset,\aloc') 
}
&
\\
& = &
\displaystyle{
\sum_{\acp' \in \cpsofreg{\regcont{\areg}{k}}} 
\weight_{\acp'} \cdot 
\anrgdist_i(\aloc',\regcont{\areg'}{k})
\; .
}
&
\end{array}
\]
\begin{flushright}
\qed
\end{flushright}
\end{proof}

The next lemma considers time-successor transitions of the region graphs 
for granularity $k$ and $2k$:
as it relies on standard non-probabilistic reasoning on the region graphs,
we omit its proof.

\begin{lemma}\label{lem:one-step_time_k_2k}
\sloppypar{
Let $k \in \Nset$
and let $(\aloc,\areg) \in \rgstates{\acdpta}{2k}$ be a state of $\acdrg{\acdpta}{2k}$.
For each transition $((\aloc,\areg),\rsuccact,\dirac{(\aloc,\areg')}) \in \rgtransdelay{\acdpta}{2k}$ 
of $\acdrg{\acdpta}{2k}$,
there exists a transition 
$(\aloc,\regcont{\areg}{k}),\rsuccact,\dirac{(\aloc',\regcont{\areg'}{k})}) \in \rgtransdelay{\acdpta}{k}$ 
of $\acdrg{\acdpta}{k}$.
}
\end{lemma}

The following lemma is an analogue of Lemma~\ref{lem:concrete_support},
applied to the case of $k$-regions and $2k$-regions.

\begin{lemma}\label{lem:k_2k_support}
Let $(\aloc,\areg) \in \regs{\maxbound}{2k}$ be a state of the region graph with granularity $2k$,
and let $((\aloc,\areg),(\acp,(\aloc,\g,\pd)),\anrgdist) \in \rgtranspedge{\acdpta}{2k}$ 
be a transition of $\acdrg{\acdpta}{2k}$.
For each pair $(\aloc_1,\areg_1), (\aloc_2,\areg_2) \in \support(\anrgdist)$ 
such that $(\aloc_1,\areg_1) \neq (\aloc_2,\areg_2)$,
we have $(\aloc_1,\regcont{\areg_1}{k}) \neq (\aloc_2,\regcont{\areg_2}{k})$.
\end{lemma}
\begin{proof}
Let $(\aloc_1,\areg_1), (\aloc_2,\areg_2) \in \support(\anrgdist)$ such that $(\aloc_1,\areg_1) \neq (\aloc_2,\areg_2)$.
If $\aloc_1 \neq \aloc_2$ then trivially
$(\aloc_1,\regcont{\areg_1}{k}) \neq (\aloc_2,\regcont{\areg_2}{k})$.
Now consider the case in which $\aloc_1 = \aloc_2$ and $\areg_1 \neq \areg_2$.
Note that $\areg_1=\areg[\aclockset_1:=0]$ and $\areg_2=\areg[\aclockset_2:=0]$.
Let $\aclockset_1'$ and $\aclockset_2'$ be the set of clocks that are equal to 0 in $\areg_1$ and $\areg_2$, respectively,
and note that $\aclockset_1' \neq \aclockset_2'$.
Then $\regcont{\areg_1}{k}
= (\anatval_1,[\aclockset_{1,0}, \aclockset_{1,1}..., \aclockset_{1,n_1}])$
and $\regcont{\areg_2}{k} 
= (\anatval_2,[\aclockset_{2,0}, \aclockset_{2,1}..., \aclockset_{2,n_2}])$
have the following properties:
either there exists clock $\aclock \in \aclockset_1' \setminus \aclockset_2'$
such that $\anatval_1(\aclock)=0$ and $\aclock \in \aclockset_{1,0}$
but either $\anatval_2(\aclock) \neq 0$ or $\aclock \not\in \aclockset_{2,0}$,
or there exists clock $\aclock \in \aclockset_2' \setminus \aclockset_1'$
such that $\anatval_2(\aclock)=0$ and $\aclock \in \aclockset_{2,0}$
but either $\anatval_1(\aclock) \neq 0$ or $\aclock \not\in \aclockset_{1,0}$.
Hence we have $(\aloc_1,\regcont{\areg_1}{k}) \neq (\aloc_2,\regcont{\areg_2}{k})$.
\qed
\end{proof}

Given $((\aloc,\areg),(\acp,(\aloc,\g,\pd)),\anrgdist) \in \rgtranspedge{\acdpta}{2k}$
and $(\aloc',\areg') \in \rgstates{\acdpta}{k}$,
Lemma~\ref{lem:k_2k_support} specifies that there exists at most one $2k$-region $\areg''$
such that $(\aloc',\areg'') \in \support(\anrgdist)$ and $\areg'' \subseteq \areg'$.
In the case in which such a $2k$-region $\areg''$ exists, we let $\areg_{\anrgdist,(\aloc',\areg')} = \areg''$,
otherwise we can set $\areg_{\anrgdist,(\aloc',\areg')}$ be equal to an arbitrary $2k$-region.
From this fact, together with Lemma~\ref{lem:one-step_pedge_k_2k} and Lemma~\ref{lem:one-step_time_k_2k},
we obtain the following lemma.
Its proof is similar to that of Lemma~\ref{lem:one-step},
and hence we omit it.

\begin{lemma}\label{lem:one-step_k_2k}
Let $(\aloc,\areg) \in \rgstates{\acdpta}{k}$ be a state of the region graph with granularity $2k$.
For each transition $((\aloc,\areg),(\acp,(\aloc,\g,\pd)),\anrgdist) \in \rgtranspedge{\acdpta}{2k}$ 
of $\acdrg{\acdpta}{2k}$,
there exists a combined transition 
$(\set{(\aloc,\regcont{\areg}{k},\anaction_i,\anrgdist_i)}{i \in I}, \set{\ctweight_i}{i \in I})$ 
of $\acdrg{\acdpta}{k}$
such that, for each $(\aloc',\areg') \in \rgstates{\acdpta}{k}$,
we have:
\begin{enumerate}
\item
$\anrgdist(\aloc',\areg_{\anrgdist,(\aloc',\areg')}) = \sum_{i \in I} \ctweight_i \cdot \anrgdist_i(\aloc',\areg')$.
\item
$\sum_{\areg'' \in \regs{\maxbound}{2k} \mbox{ s.t. } \regcont{\areg''}{k}=\areg'} \anrgdist(\aloc',\areg'') 
= 
\sum_{i \in I} \ctweight_i \cdot \anrgdist_i(\aloc',\areg')$.
\end{enumerate}
\end{lemma}

Consider equivalence $\anequiv \subseteq (\rgstates{\acdpta}{2k} \uplus \rgstates{\acdpta}{k})^2$ 
over the states of
the disjoint union of $\acdrg{\acdpta}{2k}$ and $\acdrg{\acdpta}{k}$
defined as the smallest equivalence satisfying the following conditions:
\begin{itemize}
\item
for states $(\aloc,\areg),(\aloc',\areg') \in \rgstates{\acdpta}{2k}$,
we have $(\aloc,\areg) \anequiv (\aloc',\areg')$ 
if $\aloc=\aloc'$, and $\regcont{\areg}{k}=\regcont{\areg'}{k}$
(i.e., $\areg$ and $\areg'$ are contained in the same $k$-region in $\regs{\maxbound}{k}$);
\item
for $(\aloc,\areg) \in \rgstates{\acdpta}{2k}$, $(\aloc',\areg') \in \rgstates{\acdpta}{k}$,
$(\aloc,\areg) \anequiv (\aloc',\areg')$ 
if $\aloc=\aloc'$ and $\regcont{\areg}{k} = \areg'$
(i.e., $\areg$ is contained in $\areg'$).
\end{itemize}
%
%
We then obtain the following corollary from part (2) of Lemma~\ref{lem:one-step_k_2k}.

\begin{corollary}\label{cor:one-step_k_2k}
\sloppypar{
Let $(\aloc,\areg) \in \rgstates{\acdpta}{2k}$ be a state of $\acdrg{\acdpta}{2k}$.
For each transition $((\aloc,\areg),\anaction,\anrgdist) \in \rgtrans{\acdpta}{2k}$ of $\acdrg{\acdpta}{2k}$,
there exists a combined transition 
$(\set{(\aloc,\regcont{\areg}{k}),\anaction_i,\anrgdist_i)}{i \in I}, \set{\ctweight_i}{i \in I})$ 
of $\acdrg{\acdpta}{k}$
such that
$\anrgdist 
\equiv
\distsum_{i \in I} \ctweight_i \cdot \anrgdist_i(\aloc',\areg')$,
$\anaction_i = \rsuccact$ for all $i \in I$ if $\anaction = \rsuccact$
and $\set{\anaction_i}{i \in I} \subseteq \cps{\maxbound}{k} \times \pedges$ otherwise.
}
\end{corollary}

We now proceed to the proof of part~(2) of Proposition~\ref{prop:approx}.

\begin{proof}[of part~(2) of Proposition~\ref{prop:approx}]
Consider the relation $\proposedpsim \subseteq (\rgstates{\acdpta}{2k} \uplus \rgstates{\acdpta}{k})^2$
such that $\proposedpsim$ is the smallest relation satisfying:
for $(\aloc,\areg) \in \rgstates{\acdpta}{2k}$, $(\aloc',\areg') \in \rgstates{\acdpta}{k}$,
$(\aloc,\areg) \proposedpsim (\aloc',\areg')$ 
if $(\aloc,\regcont{\areg}{k}) = (\aloc',\areg')$.
By Corollary~\ref{cor:one-step_k_2k},
we have that $\proposedpsim$ is a probabilistic simulation respecting $\anequiv$ 
and $\{ \rsuccact \}$.
Then, by Proposition~\ref{prop:hahn},
we have that:
\begin{eqnarray*}
\maxval{\acdrg{\acdpta}{2k}}{\finalregs{\maxbound}{2k}}
{\altstrategiespts{\acdrg{\acdpta}{2k}}{\{ \rsuccact \}}} 
& \leq &
\maxval{\acdrg{\acdpta}{k}}{\finalregs{\maxbound}{k}}
{\altstrategiespts{\acdrg{\acdpta}{k}}{\{ \rsuccact \}}}
\\
\minval{\acdrg{\acdpta}{2k}}{\finalregs{\maxbound}{2k}}
{\altstrategiespts{\acdrg{\acdpta}{2k}}{\{ \rsuccact \}}} 
& \geq & 
\minval{\acdrg{\acdpta}{k}}{\finalregs{\maxbound}{k}}
{\altstrategiespts{\acdrg{\acdpta}{k}}{\{ \rsuccact \}}} \; .
\end{eqnarray*}
Noting that $\rgstrategies{\acdpta}{2k} = \altstrategiespts{\acdrg{\acdpta}{2k}}{\{ \rsuccact \}}$
and $\rgstrategies{\acdpta}{k} = \altstrategiespts{\acdrg{\acdpta}{k}}{\{ \rsuccact \}}$,
we have that 
$\maxval{\acdrg{\acdpta}{2k}}{\finalregs{\maxbound}{2k}}{\rgstrategies{\acdpta}{2k}}
\leq 
\maxval{\acdrg{\acdpta}{k}}{\finalregs{\maxbound}{k}}{\rgstrategies{\acdpta}{k}}$
and
$\minval{\acdrg{\acdpta}{2k}}{\finalregs{\maxbound}{2k}}{\rgstrategies{\acdpta}{2k}}
\geq 
\minval{\acdrg{\acdpta}{k}}{\finalregs{\maxbound}{k}}{\rgstrategies{\acdpta}{k}}$.
\qed
\end{proof}

%


%
%
%
%
%